\documentclass[preprint,12pt,nopreprintline]{elsarticle}
\usepackage{amsmath,amssymb,amsfonts}
\usepackage{amsthm}
\usepackage{graphicx}
\usepackage{textcomp}
\usepackage{color}
\usepackage{epstopdf}
\usepackage{xcolor}
\usepackage{multirow}
\usepackage{subcaption}
\usepackage{multicol,lipsum}
\usepackage[utf8x]{inputenc}
\usepackage{epsfig}
\usepackage{algorithm}
\usepackage{algorithmic}
\usepackage{multirow}
\usepackage{amsmath}
\usepackage{hyperref}
\usepackage{mathtools}

\usepackage{bm}
\usepackage{listings}
\newtheorem{remark}{Remark}
\newtheorem{definition}{Definition}
\newtheorem{theorem}{Theorem}
\newtheorem{proposition}{Proposition}
\newtheorem{lemma}{Lemma}
\newtheorem{example}{Example}
\newtheorem{corollary}{Corollary}
\newtheorem{assumption}{Assumption}
\usepackage{balance}
\usepackage{textcomp}
\usepackage{xcolor}
\newcommand{\oomit}[1]{}

\journal{}
\date{}

\begin{document}

\begin{frontmatter}



\title{A New Framework for Bounding Reachability Probabilities of Continuous-time Stochastic Systems\footnote{To appear in Nonlinear Analysis: Hybrid Systems}}


\author{Bai Xue} 
\ead{xuebai@ios.ac.cn}
\affiliation{organization={Key Laboratory of System Software (Chinese Academy of Sciences) and State Key Lab. of Computer Science, Institute of Software Chinese Academy of Sciences},
            city={Beijing},
            postcode={100190}, 
            country={China}}

\begin{abstract}
This manuscript presents an innovative framework for constructing barrier functions to bound reachability probabilities for continuous-time stochastic systems described by stochastic differential equations (SDEs). The reachability probabilities considered in this paper encompass two aspects: the probability of reaching a set of specified states within a predefined finite time horizon, and the probability of reaching a set of specified states at a particular time instant. The barrier functions presented in this manuscript are developed either by relaxing a parabolic partial differential equation that characterizes the exact reachability probability or by applying the Gr\"onwall's inequality. In comparison to the prevailing construction method, which relies on Doob's non-negative supermartingale inequality (or Ville's inequality), the proposed barrier functions provide stronger alternatives, complement existing methods, or fill gaps. 
\end{abstract}


\end{frontmatter}

\section{Introduction}
\label{sec:intro}
Stochastic phenomena are commonly observed in both natural and artificial systems, spanning multiple disciplines such as biology and robotics. To accurately model these systems, sophisticated approaches are needed due to their inherent randomness \cite{franzle2008stochastic}. Stochastic differential equations (SDEs) provide a powerful tool by integrating deterministic dynamics with stochastic processes, offering a comprehensive framework for comprehending the behavior of these systems \cite{kloeden1992stochastic}. They have been widely applied, such as in models of disturbances in engineered systems like wind forces \cite{wang2015long} and pedestrian motion \cite{hoogendoorn2004pedestrian}.

The reachability probability is a critical quantitative measure within the context of SDEs \cite{lavaei2022automated}. It provides valuable insights into the likelihood of a system, governed by an SDE, reaching a set of specified states within a predetermined (in)finite time frame (referred to as reachability probability within (in)finite time horizons) or at a particular time instant (referred to as reachability probability at specific time instants). This concept plays a pivotal role in understanding the probabilistic evolution of systems under stochastic influences, enabling informed analysis and decision-making in various fields. Computing reachability probabilities typically involves solving Hamilton-Jacobi-Bellman equations \cite{koutsoukos2008computational,bujorianu2009stochastic,esfahani2016stochastic}. However, obtaining analytical solutions is often infeasible, necessitating the use of numerical approximations. As a result, obtaining both upper and lower bounds of reachability probabilities becomes impractical. \cite{nilsson2020lyapunov} gave comparison results for SDEs that via a Lyapunov-like function allow reachability probabilities within finite time horizons to be upper-bounded by an exit probability of a one-dimensional Ornstein-Uhlenbeck process, but the bounds are not in closed form. Inspired by Lyapunov functions for stability analysis, determining upper and lower bounds of reachability probabilities within infinite and finite time horizons has been simplified by finding barrier certificates.  The primary focus of safety verification studies lies in calculating upper bounds, where the objective is to estimate the maximum likelihood of reaching a specified unsafe set. Conversely, when the goal is to assess the probability of reaching a target set, the emphasis is on obtaining a lower bound of this probability. With the development of polynomial optimization, specifically sum-of-squares polynomial optimization \cite{parrilo2003semidefinite}, barrier certificates have emerged as a powerful tool for certifying upper bounds of reachability probabilities. When the system of interest is polynomial, the problem of finding barrier certificates can be addressed through convex optimization. Barrier certificates for SDEs were initially introduced in \cite{prajna2004stochastic,prajna2007framework} for infinite time safety verification, upper-bounding the probability that a system will eventually reach an unsafe region based on a non-negative barrier function. They build upon the known Doob's non-negative supermartingale inequality \cite{doob1939jean}, which requires the expectation of the non-negative barrier certificate to decrease along the system dynamics. Later, inspired by results in \cite{kushner1967stochastic} and the Doob's non-negative supermartingale inequality, \cite{steinhardt2012finite} extended barrier certificates to safety verification over finite time horizons and proposed c-martingales, which allow the expected value of the barrier function to increase over time. This approach provides upper bounds for the reachability probability of a system entering an unsafe region within finite time horizons. The c-martingales were further enhanced in \cite{santoyo2021barrier} for safety verification over finite time horizons by imposing a state-dependent bound on the expected value of the barrier certificate.  Recently, a controlled version was presented in \cite{wang2021safety}. Meanwhile, \cite{feng2020unbounded} proposed a time-varying barrier function to upper bound the reachability probability within finite time horizons, utilizing Doob's non-negative supermartingale inequality. While there has been considerable research focusing on providing upper bounds for the reachability probability within (in)finite time horizons, the practice of lower-bounding this probability has received considerably less attention. A novel equation, which can characterize the exact reachability probability within the infinite time horizon, was proposed in \cite{xue2021reach,Xue2023}. By relaxing this equation, barrier-like conditions can be obtained to both lower-bound and upper-bound the reachability probability within the infinite time horizon. Recently, this approach has been extended to lower-bound the reachability probability within finite time horizons in \cite{xue2023reachability,xue2023safe}. All of the aforementioned works study the bounding problem of reachability probabilities within (in)finite time horizons. However, to the best of our knowledge, there is no work in the framework of barrier functions investigating the problem of bounding reachability probabilities at specific time instants. This is an important problem because it addresses scenarios where system performance is critical at precise moments, not just within an interval. For instance, in multi-agent systems, agents may need to synchronize their states (e.g., achieve a rendezvous) at an exact time. Similarly, in digital control, the state is typically sampled and controlled at discrete instants, making its properties at those specific times paramount for stability and performance. Providing formal bounds for these exact-time reachability probabilities thus enables verification of stricter temporal specifications and enhances the applicability of barrier certificate methods to time-critical stochastic systems.

This paper explores the issue of lower- and upper-bounding reachability probabilities within finite time horizons and at specific time instants in stochastic systems modeled by SDEs. To tackle these problems, we propose time-dependent and time-independent barrier functions that provide lower and upper bounds for these reachability probabilities. The development of these barrier certificates is influenced by our previous work \cite{Xue2023,xue2023safe}, which introduces an alternative method that does not rely on the commonly used Doob's non-negative supermartingale inequality. Leveraging the occupation measure, the construction of the barrier certificates is achieved through either relaxation of a second-order partial differential equation or utilization of the Gr\"onwall inequality. These barrier certificates are either more powerful compared to those found in previous works, complement the existing ones, or fill a gap. They will facilitate the gain of tight bounds on reachability probabilities within finite time horizons and at specific time instants.

The main contributions of this work are summarized below.
\begin{enumerate}
\item \textbf{Novel Methodological Foundation:} The framework provides an alternative to the prevailing method that relies on Doob’s nonnegative supermartingale inequality (as used in works like [5] and [18]). Instead, it constructs barrier functions by relaxing a parabolic partial differential equation (PDE) that characterizes the exact reachability probability or by applying Grönwall’s inequality.
\item  \textbf{Comprehensive Bounding:} The work provides methods for both upper- and lower-bounding two types of reachability probabilities:
\begin{enumerate}
\item The probability of reaching a set within a finite time horizon.
\item The probability of being in a set at a specific time instant (a problem not previously addressed within the barrier function framework for SDEs).
\end{enumerate}
\item \textbf{Alternative Formulation and Filled Gaps:} A key difference from supermartingale-based approaches is that our conditions do not inherently require the barrier function to be non-negative everywhere for all time. This alternative formulation:
\begin{enumerate}
    \item Complements existing supermartingale-based methods (e.g., \cite{feng2020unbounded}, \cite{santoyo2021barrier}) and can, in practice, lead to the discovery of barrier functions that might not satisfy stricter non-negativity constraints.
    \item Fills a significant gap by providing a practical approach for lower-bounding reachability probabilities, which has received considerably less attention in the literature.
    \item Extends the reachability analysis to the novel case of probabilities at a specific time instant.
\end{enumerate}
\end{enumerate}

This paper is structured as follows. In Section \ref{sec:pre}, we introduce SDEs and the problems related to bounding reachability probabilities within finite time horizons and at specific time instants. In Section \ref{sec:bup}, we present our time-dependent and time-independent barrier functions for lower- and upper-bounding reachability probabilities within finite time horizons. Then, in Section \ref{sec:bupII}, we present our barrier functions for lower- and upper-bounding 
 reachability probabilities at specific time instants. Finally, in Section \ref{sec:con}, we conclude the paper.

Some basic notions are used throughout this paper: $\mathbb{R}$ and $\mathbb{R}_{\geq 0}$ stand for the set of real numbers and non-negative real numbers, respectively; 
$\mathbb{R}^{n}$ and $\mathbb{R}^{n\times m}$ denote the space of all $n$-dimensional vectors and $n\times m$ real matrices, respectively; for a set $\mathcal{A}$, the sets $\mathcal{A}^{\circ}$, $\overline{\mathcal{A}}$, and $\partial \mathcal{A}$ denote the interior, the closure, and the boundary of the set $\mathcal{A}$, respectively; $\wedge$ denotes the logical operation of conjunction. 
\section{Preliminaries}
\label{sec:pre}
This section introduces SDEs and the reachability probabilities bounding problem of interest. 

Consider the continuous-time stochastic system, 
\begin{equation}
    \label{SDE}
    \begin{split}
d\bm{x}(t,\bm{w})=\bm{b}(\bm{x}(t,\bm{w}))dt+\bm{\sigma}(\bm{x}(t,\bm{w})) d\bm{W}(t,\bm{w}),
\end{split}
\end{equation}
where $\bm{b}(\,\cdot\,)\colon \mathbb{R}^n\rightarrow \mathbb{R}^n$ and $\bm{\sigma}(\,\cdot\,)\colon \mathbb{R}^n \rightarrow \mathbb{R}^{n\times k}$ are locally Lipschitz continuous function; $\bm{W}(\cdot,\cdot)\colon \mathbb{R}\times \Omega\rightarrow \mathbb{R}^k$ is a standard $k$-dimensional Wiener process, and $\Omega$, equipped with the probability measure $\mathbb{P}$,  is the sample space $\bm{w}$ belongs to.  The expectation with respect to $\mathbb{P}$ is denoted by $\mathbb{E}[\,\cdot\,]$.

Given an initial state $\bm{x}_0$, the SDE \eqref{SDE} has a unique (maximal local) strong solution over a time interval $[0,T^{\bm{x}_0}(\bm{w}))$, where $T^{\bm{x}_0}(\bm{w})$ is a positive real value or infinity. This solution is denoted as $\bm{X}_{\bm{x}_0}^{\bm{w}}(\,\cdot\,)\colon [0,T^{\bm{x}_0}(\bm{w}))\rightarrow \mathbb{R}^n$, which satisfies the stochastic integral equation,
\begin{equation*}
\begin{split}
  \bm{X}_{\bm{x}_0}^{\bm{w}}(t)=\bm{x}_0+\int_{0}^t \bm{b}(\bm{X}_{\bm{x}_0}^{\bm{w}}(\tau))d \tau+\int_{0}^t \bm{\sigma}(\bm{X}_{\bm{x}_0}^{\bm{w}}(\tau)) d\bm{W}(\tau,\bm{w}).
   \end{split}
\end{equation*}

Given a  function  $v(t,\bm{x})$ that is twice continuously differentiable over $\bm{x}$ and continuously differentiable over $t$, the infinitesimal generator underlying system \eqref{SDE} on the function $v(t,\bm{x})$, which represents the limit of the expected value of $v(t,\bm{X}_{\bm{x}_0}^{\bm{w}}(t))$ as $t$ approaches 0, is presented in Definition \ref{inf_generator}.
\begin{definition}[\cite{oksendal2013stochastic}]
\label{inf_generator}
Given system \eqref{SDE},  the infinitesimal generator is the operator $\mathcal{L}$, which is defined to act on suitable functions $v(\cdot,\cdot): \mathbb{R}\times \mathbb{R}^n\rightarrow \mathbb{R}$ by
\begin{equation}
\label{infi}
    \begin{split}
        \mathcal{L}v(t,\bm{x})=\lim_{\Delta t\rightarrow 0}\frac{\mathbb{E}[v(t+\Delta t,\bm{X}_{\bm{x}}^{\bm{w}}(t+\Delta t))]-v(t,\bm{x})}{\Delta t}.
    \end{split}
\end{equation}
The domain of $\mathcal{L}$ is by definition the set of functions $v$ for which the limit \eqref{infi} exists for all $\bm{x}\in \mathbb{R}^n$ and $t\in \mathbb{R}$. 
\end{definition}

The following proposition presents the infinitesimal generator $\mathcal{L}$. 

\begin{proposition}[\cite{oksendal2013stochastic}]
Given system \eqref{SDE},  the infinitesimal generator $\mathcal{L}$ on a test function $v(t,\bm{x})$ is 
\begin{equation*}
\begin{split}
\mathcal{L}v(t,\bm{x})=\frac{\partial v(t,\bm{x}) }{\partial t}+\frac{\partial v(t,\bm{x})}{\partial \bm{x}}\bm{b}(\bm{x})+\frac{1}{2}\textbf{\rm tr}(\bm{\sigma}(\bm{x})^{\top}\frac{\partial^2 v(t,\bm{x})}{\partial \bm{x}^2} \bm{\sigma}(\bm{x})),
\end{split}
\end{equation*}
where $\frac{\partial v}{\partial t}$ and $\frac{\partial v}{\partial \bm{x}}$ represent the gradient of the test function $v(t,\bm{x})$ with respect to $t$ and $\bm{x}$, respectively, $\frac{\partial^2 v}{\partial \bm{x}^2}$ represents the second-order partial derivative of the test function $v(t,\bm{x})$ with respect to $\bm{x}$, and $\textbf{\rm tr}(\,\cdot\,)$ denotes the trace of a matrix.

The domain of $\mathcal{L}$ contains the set of functions $v$ of compact support, which are twice continuously differentiable over $\bm{x}$ and continuously differentiable over $t$.  
\end{proposition}

Given a state constrained set $\mathcal{X}\subseteq \mathbb{R}^n$ that is open and bounded, and a subset $\mathcal{X}_s\subseteq \mathcal{X}$ that is closed, the reachability probability within a time horizon $[0,T]$ with $0<T<\infty$ is the probability of  system \eqref{SDE}, starting from an initial state $\bm{x}_0\in \mathcal{X}\setminus \mathcal{X}_s$, reaches the set $\mathcal{X}_s$ within the time horizon $[0,T]$ while remaining within the state constrained set $\mathcal{X}$ until the first occurrence of hitting the set $\mathcal{X}_s$. It is formulated in Definition \ref{un_pro}.

 \begin{definition}[Reachability Probability I]
\label{un_pro}
Given a time horizon $[0,T]$ with $0<T<\infty$ and an initial state $\bm{x}_0\in \mathcal{X}\setminus \mathcal{X}_s$, the reachability probability $\mathbb{P}_{\bm{x}_0}^{[0,T]}$ within the time horizon $[0,T]$ for system \eqref{SDE} starting from an initial state $\bm{x}_0\in \mathcal{X}\setminus \mathcal{X}_s$ is defined below: 
\begin{equation*}
\mathbb{P}_{\bm{x}_0}^{[0,T]}: =\mathbb{P}\left(\left\{\bm{w}\in \Omega \,\middle|\,
\begin{aligned}
\exists t\in [0,T]. \bm{X}_{\bm{x}_0}^{\bm{w}}(t) \in \mathcal{X}_s\wedge \forall \tau \in [0,t). \bm{X}_{\bm{x}_0}^{\bm{w}}(\tau) \in \mathcal{X}
\end{aligned}
\right\}
\right).
\end{equation*}
\end{definition}

For continuous-time systems with continuous sample paths and for any initial state $\boldsymbol{x}_0 \in \mathcal{X} \setminus \mathcal{X}_s$, the event formulation in Definition~2,
\[
\exists t \in [0,T]: \boldsymbol{X}_{\boldsymbol{x}_0}^{\boldsymbol{w}}(t) \in \mathcal{X}_s \;\wedge\; \forall \tau \in [0,t): \boldsymbol{X}_{\boldsymbol{x}_0}^{\boldsymbol{w}}(\tau) \in \mathcal{X},
\]
is equivalent to the formulation $
\exists t \in [0,T]: \boldsymbol{X}_{\boldsymbol{x}_0}^{\boldsymbol{w}}(t) \in \mathcal{X}_s \;\wedge\; \forall \tau \in [0,t): \boldsymbol{X}_{\boldsymbol{x}_0}^{\boldsymbol{w}}(\tau) \in \mathcal{X} \setminus \mathcal{X}_s$. This equivalence holds because path continuity ensures the existence of a first hitting time $t$ for the set $\mathcal{X}_s$. Consequently, the condition $\forall \tau \in [0, t)\, \boldsymbol{X}_{\boldsymbol{x}_0}^{\boldsymbol{w}}(\tau) \in \mathcal{X}$ implicitly requires the trajectory to remain in $\mathcal{X} \setminus \mathcal{X}_s$ for all times before $t$.

Given a state constrained set $\mathcal{X}\subseteq \mathbb{R}^n$ that is open and bounded, and a subset $\mathcal{X}_s\subseteq \mathcal{X}$ that is closed, the reachability probability at a time instant $T$ with $0<T<\infty$ is the probability of  system \eqref{SDE}, starting from an initial state $\bm{x}_0\in \mathcal{X}\setminus \mathcal{X}_s$, reaches the set $\mathcal{X}_s$ at the time instant $T$ while remaining within the state constrained set $\mathcal{X}$ before the time $T$. It is formulated in Definition \ref{un_proII}.

 \begin{definition}[Reachability Probability II]
\label{un_proII}
Given a time horizon $[0,T]$ with $0<T<\infty$ and an initial state $\bm{x}_0\in \mathcal{X}\setminus \mathcal{X}_s$, the reachability probability $\mathbb{P}_{\bm{x}_0}^{T}$ at the time instant $T>0$ for  system \eqref{SDE} starting from an initial state $\bm{x}_0\in \mathcal{X}\setminus \mathcal{X}_s$ is defined below: 
\begin{equation*}
\mathbb{P}_{\bm{x}_0}^T: =\mathbb{P}\big(\left\{\bm{w}\in \Omega \,\middle|\,
\begin{aligned}
\bm{X}_{\bm{x}_0}^{\bm{w}}(T) \in \mathcal{X}_s\wedge\forall \tau \in [0,T). \bm{X}_{\bm{x}_0}^{\bm{w}}(\tau) \in \mathcal{X}
\end{aligned}
\right\}
\big).
\end{equation*}
\end{definition}

In this paper, we address the challenge of computing the exact reachability probabilities in Definition \ref{un_pro} and \ref{un_proII}, which is often infeasible for nonlinear systems. Instead, we resort to characterizing their lower and upper bounds, i.e., we will characterize $\delta_{i,1} \in [0,1]$ and $\delta_{i,2} \in [0,1]$, $i=1,2$, such that 
\[\delta_{1,1}\leq \mathbb{P}_{\bm{x}_0}^{[0,T]} \leq \delta_{1,2}\]
and 
\[\delta_{2,1}\leq \mathbb{P}_{\bm{x}_0}^{T} \leq \delta_{2,2}.\]

\section{Bounding Reachability Probabilities I}
\label{sec:bup}
This section introduces our barrier functions for upper- and lower-bounding the reachability probability $\mathbb{P}_{\bm{x}_0}^{[0,T]}$ in Definition \ref{un_pro}.  They are respectively formulated in Subsection \ref{subsec:uup} and \ref{subsec:lup}.

The construction of these barrier functions lies on an auxiliary stochastic process $\{\widehat{\bm{X}}_{\bm{x}_0}^{\bm{w}}(t), t\in \mathbb{R}_{\geq 0}\}$ for $\bm{x}_0\in \overline{\mathcal{X}\setminus \mathcal{X}_s}$ that is a stopped process corresponding to $\{\bm{X}_{\bm{x}_0}^{\bm{w}}(t), t\in [0,T^{\bm{x}_0}(\bm{w}))\}$ and the set $\overline{\mathcal{X}\setminus \mathcal{X}_s}$, i.e., 
\begin{equation}
\widehat{\bm{X}}_{\bm{x}_0}^{\bm{w}}(t)=
\begin{cases}
\bm{X}_{\bm{x}_0}^{\bm{w}}(t), & \text{if $t<\tau^{\bm{x}_0}(\bm{w})$},\\
\bm{X}_{\bm{x}_0}^{\bm{w}}(\tau^{\bm{x}_0}(\bm{w})),  & \text{if $t\geq \tau^{\bm{x}_0}(\bm{w})$},
\end{cases}
\end{equation}
where $\tau^{\bm{x}_0}(\bm{w})=\inf\{\,t\mid \bm{X}_{\bm{x}_0}^{\bm{w}}(t) \in \partial \mathcal{X}\cup \partial \mathcal{X}_s\,\}$ is the first time of exit of $\bm{X}_{\bm{x}_0}^{\bm{w}}(t)$ from the set $\mathcal{X}\setminus \mathcal{X}_s$. It is worth remarking that the first time of exit $\tau^{\boldsymbol{x}_0}(\boldsymbol{w})$ is well-defined and $\tau^{\boldsymbol{x}_0}(\boldsymbol{w}) \leq T^{\boldsymbol{x}_0}(\boldsymbol{w})$. This holds because the state constrained set $\mathcal{X}$ is bounded. Consequently, if the path $\boldsymbol{X}_{\boldsymbol{x}_0}^{\boldsymbol{w}}(t)$ has a finite explosion time $T^{\boldsymbol{x}_0}(\boldsymbol{w}) < \infty$ (i.e., it escapes to infinity), it must first exit the bounded domain $\mathcal{X} \setminus \mathcal{X}_s$. Exiting this domain necessitates touching its boundary, $\partial\mathcal{X} \cup \partial\mathcal{X}_s$, at a time prior to the explosion.\oomit{It is worth remarking here that if the path $\bm{X}_{\bm{x}_0}^{\bm{w}}(t)$ escapes to infinity in finite time, it must touch $\partial \mathcal{X}\cup \partial \mathcal{X}_s$ and thus $\tau^{\bm{x}_0}(\bm{w})\leq T^{\bm{x}_0}(\bm{w})$.}The stopped process $\widehat{\bm{X}}_{\bm{x}_0}^{\bm{w}}(t)$ inherits the right continuity and strong Markovian property of $\bm{X}_{\bm{x}_0}^{\bm{w}}(t)$. Moreover, the infinitesimal generator $\widehat{\mathcal{L}}$ on a test function $v(t,\bm{x})$ corresponding to $\widehat{\bm{X}}_{\bm{x}_0}^{\bm{w}}(t)$ is identical to the one corresponding to $\bm{X}_{\bm{x}_0}^{\bm{w}}(t)$ over $\mathcal{X}\setminus \mathcal{X}_s$, and is equal to $\frac{\partial v(t,\bm{x})}{\partial t}$ on $\partial \mathcal{X}\cup \partial \mathcal{X}_s$ \cite{kushner1967stochastic}. That is, for $v(t,\bm{x})$ that is twice continuously differentiable over $\bm{x}$ and continuously differentiable over $t$, 
\[
\begin{split}
\widehat{\mathcal{L}}v(t,\bm{x})=\mathcal{L}v(t,\bm{x})=&\frac{\partial v(t,\bm{x})}{\partial t}+\frac{\partial v(t,\bm{x})}{\partial \bm{x}}\bm{b}(\bm{x})+\frac{1}{2}\textbf{\rm tr}(\bm{\sigma}(\bm{x})^{\top}\frac{\partial^2 v(t,\bm{x})}{\partial \bm{x}^2} \bm{\sigma}(\bm{x}))
\end{split}
\] for $(\bm{x},t)\in \mathcal{X}\setminus \mathcal{X}_s \times [0,T]$ and \[\widehat{\mathcal{L}}v(t,\bm{x})=\frac{\partial v(t,\bm{x})}{\partial t}\] for $\bm{x}\in \partial \mathcal{X}\cup \partial \mathcal{X}_s$ and $t\in [0,T]$.

Given an initial state $\bm{x}_0\in \mathcal{X}\setminus \mathcal{X}_s$, the exact reachability probability $\mathbb{P}_{\bm{x}_0}^{[0,T]}$ is equal to the probability of reaching the set $\partial \mathcal{X}_s$ at the time instant $T$ for the above auxiliary stochastic process. Before justifying this statement, we first show that the reachability probability $\mathbb{P}_{\bm{x}_0}^{[0,T]}$ is equal to the probability that system \eqref{SDE} will reach the boundary $\partial \mathcal{X}_s$ of the set $\mathcal{X}_s$ within the time horizon $[0,T]$ while remaining within the state constrained set $\mathcal{X}$ until the first occurrence of hitting the set $\partial \mathcal{X}_s$.  

\begin{lemma}
Given $\bm{x}_0\in \mathcal{X}\setminus \mathcal{X}_s$, 
\begin{equation*}
\mathbb{P}_{\bm{x}_0}^{[0,T]}=\mathbb{P}\Bigg(\left\{\bm{w}\in \Omega \,\middle|\,
\begin{aligned}
&\exists t\in [0,T]. \bm{X}_{\bm{x}_0}^{\bm{w}}(t) \in \partial \mathcal{X}_s\wedge\\
&\forall \tau \in [0,t). \bm{X}_{\bm{x}_0}^{\bm{w}}(\tau) \in \mathcal{X}\setminus \mathcal{X}_s
\end{aligned}
\right\}
\Bigg).
\end{equation*}
\end{lemma}
\begin{proof}
Since $\mathcal{X}_s$ is a closed set and the sample path $\bm{X}_{\bm{x}_0}^{\bm{w}}(t)$ is continuous over $t$, we have $A=B$, where $A=\{\,\bm{w}\in \Omega \mid \exists t\in [0,T]. \bm{X}_{\bm{x}_0}^{\bm{w}}(t) \in \partial \mathcal{X}_s\bigwedge \forall \tau \in [0,t). \bm{X}_{\bm{x}_0}^{\bm{w}}(\tau) \in \mathcal{X}\,\}$ and $B=\{\,\bm{w}\in \Omega \mid \exists t\in [0,T]. \bm{X}_{\bm{x}_0}^{\bm{w}}(t) \in \mathcal{X}_s\bigwedge \forall \tau \in [0,t). \bm{X}_{\bm{x}_0}^{\bm{w}}(\tau) \in \mathcal{X}\,\}$. Thus, we have the conclusion. 
\end{proof}

\begin{lemma}[Lemma 1, \cite{xue2023safe}]
\label{equiv}
    Given $\bm{x}_0\in \mathcal{X}\setminus \mathcal{X}_s$, \[
    \begin{split}
   & \mathbb{P}_{\bm{x}_0}^{[0,T]}=\mathbb{P}( \widehat{\bm{X}}_{\bm{x}_0}^{\bm{w}}(T)\in \partial \mathcal{X}_s)=\mathbb{E}[1_{\partial \mathcal{X}_s}(\widehat{\bm{X}}_{\bm{x}_0}^{\bm{w}}(T))].
 \end{split}
 \] Moreover, for any $0<T_1\leq T_2\leq T$, \[\mathbb{P}(\widehat{\bm{X}}_{\bm{x}_0}^{\bm{w}}(T_1)\in \partial \mathcal{X}_s)\leq \mathbb{P}(\widehat{\bm{X}}_{\bm{x}_0}^{\bm{w}}(T_2)\in \partial \mathcal{X}_s).\] 
\end{lemma}

Further, the exact reachability probability $\mathbb{P}_{\bm{x}_0}^{[0,T]}$ can be reduced to a solution to a second-order partial differential equation. 
\begin{proposition}
\label{eq}
    Suppose there exists a function $v(t,\bm{x})\colon [0,T]\times \mathbb{R}^n$ which is twice continuously differentiable over $\bm{x}$ and continuously differentiable over $t$, satisfying 
    \begin{equation}
    \label{eq_pro}
        \begin{cases}
            \widehat{\mathcal{L}}v(t,\bm{x})=0, & \forall \bm{x}\in \overline{\mathcal{X}\setminus \mathcal{X}_s},\forall t \in [0,T],\\
            v(T,\bm{x})=1_{\partial \mathcal{X}_s}(\bm{x}), & \forall \bm{x}\in \overline{\mathcal{X}\setminus \mathcal{X}_s},
        \end{cases}
    \end{equation}
then
$\mathbb{P}_{\bm{x}_0}^{[0,T]}=v(0,\bm{x}_0)$ for $\bm{x}_0\in \mathcal{X}\setminus \mathcal{X}_s$.    
\end{proposition}
\begin{proof}
The proof relies on  Dynkin’s formula and Lemma \ref{equiv}: for $\bm{x}_0\in \mathcal{X}\setminus \mathcal{X}_s$, we obtain
\[
    \begin{split}
        \mathbb{P}_{\bm{x}_0}^{[0,T]}&=\mathbb{P}(\widehat{\bm{X}}_{\bm{x}_0}^{\bm{w}}(T)\in \partial \mathcal{X}_s)=\mathbb{E}[1_{\partial \mathcal{X}_s}(\widehat{\bm{X}}_{\bm{x}_0}^{\bm{w}}(T))]
        =\mathbb{E}[v(T,\widehat{\bm{X}}_{\bm{x}_0}^{\bm{w}}(T))]\\
        &=v(0,\bm{x}_0)+\mathbb{E}[\int_{0}^T \widehat{\mathcal{L}}v(t,\widehat{\bm{X}}_{\bm{x}_0}^{\bm{w}}(t))dt]=v(0,\bm{x}_0).
    \end{split}
\]
    The proof is completed. 
\end{proof}

\subsection{Upper-bounding Reachability Probabilities}
\label{subsec:uup}
In this subsection, we present our barrier functions for upper-bounding the reachability probability, denoted as $\mathbb{P}_{\bm{x}_0}^{[0,T]}$ in Definition \ref{un_pro}. The first time-dependent barrier function is obtained via relaxing equation \eqref{eq_pro} as stated in Proposition \ref{eq}. The second one extends upon the first one, which relaxes the supermartingale requirement. The third one is a variant of the second one, using a time-independent function $v(\bm{x})$ instead of a time-dependent function $v(t,\bm{x})$. They are respectively formulated in Corollary \ref{coro1}, Theorem  \ref{coro2}, and Corollary \ref{coro3}.

\begin{corollary}
\label{coro1}
    Suppose there exists a barrier function $v(t,\bm{x})\colon [0,T]\times \mathbb{R}^n\rightarrow \mathbb{R}$ that is continuously differentiable over $t$ and twice continuously differentiable over $\bm{x}$, satisfying 
   \begin{equation}
   \label{lower_bound_condition1}
        \begin{cases}
            \mathcal{L}v(t,\bm{x})\leq 0, &\forall \bm{x}\in \mathcal{X}\setminus \mathcal{X}_s, \forall t \in [0,T], \\
            \frac{\partial v(t,\bm{x})}{\partial t}\leq 0, & \forall \bm{x}\in \partial \mathcal{X}\cup \partial \mathcal{X}_s, \forall t\in [0,T],\\ 
            v(T,\bm{x})\geq 1_{\partial \mathcal{X}_s}(\bm{x}), & \forall \bm{x}\in \overline{\mathcal{X}\setminus \mathcal{X}_s},
        \end{cases}
    \end{equation}
then $\mathbb{P}_{\bm{x}_0}^{[0,T]}\leq v(0,\bm{x}_0)$ for $\bm{x}_0\in \mathcal{X}\setminus \mathcal{X}_s$.    
\end{corollary}
\begin{proof}
The first condition in \eqref{lower_bound_condition1}, $\mathcal{L}v(t,\bm{x}) \leq 0$ for $\bm{x} \in \mathcal{X} \setminus \mathcal{X}_s$, directly implies $\widehat{\mathcal{L}}v(t,\bm{x}) \leq 0$ for $\bm{x} \in \mathcal{X} \setminus \mathcal{X}_s$.
The second condition in \eqref{lower_bound_condition1}, $\frac{\partial v(t,\bm{x})}{\partial t} \leq 0$ for $\bm{x} \in \partial\mathcal{X} \cup \partial\mathcal{X}_s$, directly implies $\widehat{\mathcal{L}}v(t,\bm{x}) \leq 0$ for $\bm{x} \in \partial\mathcal{X} \cup \partial\mathcal{X}_s$. Combining these two results with the third condition $v(T,\bm{x})\geq 1_{\partial \mathcal{X}_s}(\bm{x}),  \forall \bm{x}\in \overline{\mathcal{X}\setminus \mathcal{X}_s}$, we have 
 \begin{equation}
 \label{lower_bound_condition11}
        \begin{cases}
            \widehat{\mathcal{L}}v(t,\bm{x})\leq 0, & \forall \bm{x}\in \overline{\mathcal{X}\setminus \mathcal{X}_s},\forall t \in [0,T],\\
            v(T,\bm{x})\geq 1_{\partial \mathcal{X}_s}(\bm{x}), & \forall \bm{x}\in \overline{\mathcal{X}\setminus \mathcal{X}_s}.
        \end{cases}
    \end{equation}
  
  Based on \eqref{lower_bound_condition11}, we can obtain the conclusion using the Dynkin’s formula and Lemma \ref{equiv}: for $\bm{x}_0\in \mathcal{X}\setminus \mathcal{X}_s$, we obtain
\begin{equation*}
\begin{split}
    \mathbb{P}_{\bm{x}_0}^{[0,T]}&=\mathbb{P}( \widehat{\bm{X}}_{\bm{x}_0}^{\bm{w}}(T)\in \partial \mathcal{X}_s)=\mathbb{E}[1_{\partial \mathcal{X}_s}(\widehat{\bm{X}}_{\bm{x}_0}^{\bm{w}}(T))]\leq \mathbb{E}[v(T,\widehat{\bm{X}}_{\bm{x}_0}^{\bm{w}}(T))]\\
    &=v(0,\bm{x}_0)+\mathbb{E}[\int_{0}^T \widehat{\mathcal{L}}v(t,\widehat{\bm{X}}_{\bm{x}_0}^{\bm{w}}(t))dt]\leq v(0,\bm{x}_0).
    \end{split}
    \end{equation*}
    The proof is completed. 
\end{proof}

We found another time-dependent barrier function for upper-bounding the reachability probability $\mathbb{P}_{\bm{x}_0}^{[0,T]}$ in Theorem 5 in \cite{feng2020unbounded}, which is formulated below:
Suppose there exists a constant $\eta>0$ and a barrier function $v(t,\bm{x})\colon \mathbb{R}\times \mathbb{R}^n$, satisfying 
   \begin{equation}
   \label{fen2020}
        \begin{cases}
            \mathcal{L}v(t,\bm{x})\leq 0, & \forall \bm{x}\in \mathcal{X}\setminus \mathcal{X}_s,\forall t \in [0,T],\\
            \frac{\partial v(t,\bm{x})}{\partial t}\leq 0,&\forall \bm{x}\in \partial \mathcal{X}, \forall t\in [0,T],\\ 
            v(t,\bm{x})\geq \eta 1_{\mathcal{X}_s} (\bm{x}), &\forall \bm{x}\in \overline{\mathcal{X}}, \forall t\in [0,T],
        \end{cases}
    \end{equation}
 then $\mathbb{P}_{\bm{x}_0}^{[0,T]}\leq \frac{v(0,\bm{x}_0)}{\eta}$.   

Upon comparing conditions \eqref{lower_bound_condition1} and \eqref{fen2020}, it becomes apparent that condition \eqref{fen2020} imposes the requirement of non-negativity for the barrier function $v(t,\bm{x})$ over $[0,T]\times \overline{\mathcal{X}}$. Conversely, condition \eqref{lower_bound_condition1} solely necessitates non-negativity for the function $v(t,\bm{x})$ over $\overline{\mathcal{X}\setminus \mathcal{X}_s}$ at $t=T$. These disparities arise due to the construction of \eqref{fen2020} using the well-established Doob’s non-negative supermartingale inequality (also known as Ville’s inequality \cite{doob1939jean}). In contrast, condition \eqref{lower_bound_condition1} is formulated by relaxing equation \eqref{eq_pro}.

Corollary \ref{coro1} states that if there exists a barrier function $v(t,\bm{x})\colon [0,T]\times \mathbb{R}^n\rightarrow \mathbb{R}$ satisfying \eqref{lower_bound_condition1}, then the reachability probability $\mathbb{P}_{\bm{x}_0}^{[0,T]}$ can be bounded above by $v(0,\bm{x}_0)$. However, the requirement for $\mathcal{L}v(t,\bm{x})\leq 0$ to hold over $(t,\bm{x}) \in [0,T]\times (\mathcal{X}\setminus \mathcal{X}_s)$ may hinder the acquisition of such a barrier function. In the following, we will further relax this supermartingale requirement.

\begin{theorem}
\label{coro2}
    Suppose there exists a barrier function $v(t,\bm{x})\colon [0,T]\times \mathbb{R}^n\rightarrow \mathbb{R}$ that is continuously differentiable over $t$ and twice continuously differentiable over $\bm{x}$, satisfying 
 \begin{equation}
 \label{lower_bound_condition2}
     \begin{cases}
        \mathcal{L}v(t,\bm{x})\leq \alpha v(t,\bm{x})+\beta, & \forall \bm{x}\in \mathcal{X}\setminus~\mathcal{X}_s, \forall t\in [0,T],\\
        \frac{\partial v(t,\bm{x})}{\partial t} \leq \alpha v(t,\bm{x})+\beta, & \forall \bm{x}\in \partial \mathcal{X} \cup \partial \mathcal{X}_s, \forall t\in [0,T],\\
        v(T,\bm{x})\geq 1_{\partial \mathcal{X}_s}(\bm{x}),& \forall \bm{x}\in \overline{\mathcal{X}\setminus \mathcal{X}_s},
    \end{cases}
\end{equation}   
then
\begin{equation*}
\mathbb{P}_{\bm{x}_0}^{[0,T]}\leq
    \begin{cases}
       v(0,\bm{x}_0)+\beta T, &  \text{if $\alpha=0$},\\
       e^{\alpha T}v(0,\bm{x}_0)+\frac{\beta}{\alpha}(e^{\alpha T}-1), &\text{if $\alpha\neq 0$}
    \end{cases}
\end{equation*}
 for $\bm{x}_0\in \mathcal{X}\setminus \mathcal{X}_s$.    
\end{theorem}
\begin{proof}
From \eqref{lower_bound_condition2}, we have 
 \begin{equation*}
        \begin{cases}
            \widehat{\mathcal{L}}v(t,\bm{x})\leq \alpha v(t,\bm{x})+\beta, & \forall \bm{x}\in \overline{\mathcal{X}\setminus \mathcal{X}_s}, \forall t \in [0,T],\\
            v(T,\bm{x})\geq 1_{\partial \mathcal{X}_s}(\bm{x}), & \forall \bm{x}\in \overline{\mathcal{X}\setminus \mathcal{X}_s}.
        \end{cases}
    \end{equation*}

If $\alpha=0$, we have that 
\[
\begin{split}
\mathbb{P}_{\bm{x}_0}^{[0,T]}&=\mathbb{P}(\widehat{\bm{X}}_{\bm{x}_0}^{\bm{w}}(T)\in \partial \mathcal{X}_s)=\mathbb{E}[1_{\partial \mathcal{X}_s}(\widehat{\bm{X}}_{\bm{x}_0}^{\bm{w}}(T))]\\
&\leq \mathbb{E}[v(T,\widehat{\bm{X}}_{\bm{x}_0}^{\bm{w}}(T))]=v(0,\bm{x}_0)+\mathbb{E}[\int_{0}^T \widehat{\mathcal{L}}v(t,\widehat{\bm{X}}_{\bm{x}_0}^{\bm{w}}(t))dt]\\
&\leq v(0,\bm{x}_0)+\beta T.
\end{split}
\]

When $\alpha\neq 0$, we have, by the Gr\"onwall's inequality in the differential form, that 
\[\mathbb{E}[v(T,\widehat{\bm{X}}_{\bm{x}_0}^{\bm{w}}(T))]\leq e^{\alpha T}v(0,\bm{x}_0)+\frac{\beta}{\alpha}(e^{\alpha T}-1).\]
Also, since $v(T,\bm{x})\geq 1_{\partial \mathcal{X}_s}(\bm{x}), \forall \bm{x}\in \overline{\mathcal{X}\setminus \mathcal{X}_s}$, we have $\mathbb{P}_{\bm{x}_0}^{[0,T]}=\mathbb{E}[1_{\partial \mathcal{X}_s}(\widehat{\bm{X}}_{\bm{x}_0}^{\bm{w}}(T))]\leq \mathbb{E}[v(T,\widehat{\bm{X}}_{\bm{x}_0}^{\bm{w}}(T))]$ and thus $\mathbb{P}_{\bm{x}_0}^{[0,T]}=\mathbb{P}( \widehat{\bm{X}}_{\bm{x}_0}^{\bm{w}}(T)\in \partial \mathcal{X}_s)\leq e^{\alpha T}v(0,\bm{x}_0)+\frac{\beta}{\alpha}(e^{\alpha T}-1)$.    

The proof is completed. 
\end{proof}

It is easy to observe that condition \eqref{lower_bound_condition1} in Corollary \ref{coro1} is a special case of the one \eqref{lower_bound_condition2} in Theorem \ref{coro2} with $\alpha=\beta=0$. Below, a straightforward variant of the result in Theorem \ref{coro2} is presented. Instead of resorting to a time-dependent barrier function $v(t,\bm{x})\colon [0,T]\times \mathbb{R}^n \rightarrow \mathbb{R}$, the focus is on finding a time-independent one $v(\bm{x})\colon \mathbb{R}^n \rightarrow \mathbb{R}$.

\begin{corollary}
\label{coro3}
    Suppose there exists a barrier function $v(\bm{x})\colon \mathbb{R}^n\rightarrow \mathbb{R}$, which is twice continuously differentiable, satisfying 
      \begin{equation}
    \label{lower_bound_3}
        \begin{cases}
            \mathcal{L}v(\bm{x})\leq \alpha v(\bm{x})+\beta, &\forall \bm{x}\in \mathcal{X}\setminus \mathcal{X}_s, \\
            0\leq \alpha v(\bm{x})+\beta, & \forall \bm{x}\in \partial \mathcal{X}\cup \partial \mathcal{X}_s,\\
            v(\bm{x})\geq 1_{\partial \mathcal{X}_s}(\bm{x}),& \forall \bm{x}\in \overline{\mathcal{X}\setminus \mathcal{X}_s},
        \end{cases}
    \end{equation} 
then
\begin{equation*}
\mathbb{P}_{\bm{x}_0}^{[0,T]}\leq
    \begin{cases}
      v(\bm{x}_0)+\beta T, & \text{if $\alpha=0$},\\
      e^{\alpha T}v(\bm{x}_0)+\frac{\beta}{\alpha}(e^{\alpha T}-1), & \text{if  $\alpha\neq 0$},
    \end{cases}
\end{equation*}
 for $\bm{x}_0\in \mathcal{X}\setminus \mathcal{X}_s$.    
\end{corollary}

From $v(\bm{x})\geq 1_{\partial \mathcal{X}_s}(\bm{x}), \forall \bm{x}\in \overline{\mathcal{X}\setminus \mathcal{X}_s}$ in \eqref{lower_bound_3}, we can obtain $v(\bm{x})\geq 0$ for $\bm{x}\in \partial \mathcal{X}$ and $v(\bm{x})\geq 1$ for $\bm{x}\in \partial \mathcal{X}_s$. Since $0\leq \alpha v(\bm{x})+\beta, \forall \bm{x}\in \partial \mathcal{X}\cup \partial \mathcal{X}_s$, we can obtain $\alpha+\beta\geq 0$ if $\alpha\leq 0$. However, it is worth remarking here that $\alpha >0$ is permitted in condition \eqref{lower_bound_3}.

Another upper bound of $\mathbb{P}_{\bm{x}_0}^{[0,T]}$ was derived in \cite{santoyo2021barrier}: if $v(\bm{x})$ satisfies 
  \begin{equation}
    \label{lower_bound_4}
        \begin{cases}
            \mathcal{L}v(\bm{x})\leq \alpha v(\bm{x})+\beta, & \forall \bm{x}\in \mathcal{X}\setminus \mathcal{X}_s, \\
            v(\bm{x})\geq 1, & \forall \bm{x}\in \mathcal{X}_s,\\
            v(\bm{x})\geq 0, & \forall \bm{x}\in \mathcal{X},
        \end{cases}
    \end{equation} 
an upper bound of $\mathbb{P}_{\bm{x}_0}^{[0,T]}$ is 
   \begin{equation*}
       \begin{cases}
       (v(\bm{x}_0)-(e^{\beta T}-1)\frac{\beta}{\alpha})e^{-\beta T}, & \text{if $\alpha<0 \wedge \alpha+\beta>0$},\\
       v(\bm{x}_0)+\beta T, & \text{if $\alpha=0 \wedge \beta\geq 0$},\\
       e^{-\beta T}(v(\bm{x}_0)-1) +1,&    \text{if $\alpha<0 \wedge \alpha+\beta\leq 0 \wedge \beta\geq 0$}.
       \end{cases}
   \end{equation*}
These results are obtained via following Theorem 1 in Chapter 3 in \cite{kushner1967stochastic} and are built upon the known
Doob’s nonnegative supermartingale inequality (or, Ville’s
inequality \cite{doob1939jean}) as condition \eqref{fen2020}. However, it is observed that as $\alpha$ approaches $0^-$, $(v(\bm{x}_0)-(e^{\beta T}-1)\frac{\beta}{\alpha})e^{-\beta T}$ is not equal to $v(\bm{x}_0)+\beta T$ as expected. In contrast, it tends to infinity, which is overly conservative. Condition \eqref{lower_bound_3} in Corollary \ref{coro3} has certain advantages over \eqref{lower_bound_4}: firstly, as $\alpha$ approaches $0^-$, the expression $e^{\alpha T}v(\bm{x}_0)+\frac{\beta}{\alpha}(e^{\alpha T}-1)$ converges to $v(\bm{x}_0)+\beta T$; secondly, when $v(\bm{x}_0)\leq 1<-\frac{\beta}{\alpha}$ (when $\alpha<0$, $1<-\frac{\beta}{\alpha}$ implies $\alpha+\beta>0$), we can obtain $(v(\bm{x}_0)-(e^{\beta T}-1)\frac{\beta}{\alpha})e^{-\beta T}>e^{\alpha T}v(\bm{x}_0)+\frac{\beta}{\alpha}(e^{\alpha T}-1)$. On the other hand, it can be observed that, when $v(\bm{x}_0) > 1$, $\alpha < 0$, $\alpha + \beta > 0$, and $T > 0$, the upper bound $
\bigl(v(\bm{x}_0) - (e^{\beta T}-1)\tfrac{\beta}{\alpha}\bigr) e^{-\beta T}$ is always greater than 1 and is therefore meaningless. Consequently, under the condition $\alpha < 0$ and $\alpha + \beta > 0$, a constraint $v(\bm{x}_0)\leq 1$ should be added in \eqref{lower_bound_4} in order to make the upper bound $\bigl(v(\bm{x}_0) - (e^{\beta T}-1)\frac{\beta}{\alpha}\bigr) e^{-\beta T}$ be meaningful (i.e., be less than or equal to 1); fourthly, in the case $\alpha=0$, condition \eqref{lower_bound_3} is strictly weaker than \eqref{lower_bound_4}, because it relaxes the requirements that $v(\bm{x})$ remain nonnegative throughout $\mathcal{X}$ and that $v(\bm{x}) \geq 1$ hold over the entire $\mathcal{X}_s$. Finally, unlike the aforementioned condition \eqref{lower_bound_4}, condition \eqref{lower_bound_3} is not restricted to the scenario where $\alpha\leq 0$ and/or $\beta\geq 0$.

\subsection{Lower-bounding Reachability Probabilities}
\label{subsec:lup}
In this subsection, we present our barrier functions for lower-bounding the reachability probability $\mathbb{P}_{\bm{x}_0}^{[0,T]}$.

The construction of the first barrier function was inspired by  \cite{xue2021reach,Xue2023}. It cannot be obtained via relaxing equation \eqref{eq_pro} directly. An auxiliary function is introduced. 
\begin{proposition}
\label{coro4}
    Suppose there exist a barrier function $v(t,\bm{x})\colon [0,T]\times \mathbb{R}^n\rightarrow \mathbb{R}$ and a function $w(t,\bm{x})\colon [0,T] \times \mathbb{R}^n \rightarrow \mathbb{R}$ with $\sup_{(t,\bm{x})\in [0,T]\times \overline{\mathcal{X}}}|w(t,\bm{x})|\leq M$ that are continuously differentiable over $t$ and twice continuously differentiable over $\bm{x}$, satisfying 
   \begin{equation}
   \label{upper_bound_1}
        \begin{cases}
            \mathcal{L}v(t,\bm{x})\geq 0, & \forall \bm{x}\in \mathcal{X}\setminus \mathcal{X}_s,\forall t \in [0,T],\\
            \frac{\partial v(t,\bm{x})}{\partial t}\geq 0, & \forall \bm{x}\in \partial \mathcal{X}\cup \partial \mathcal{X}_s, \forall t\in [0,T],\\ 
            v(t,\bm{x})\leq 1+\frac{\partial w(t,\bm{x})}{\partial t},& \forall \bm{x}\in \partial \mathcal{X}_s, \forall t\in [0,T],\\
            v(t,\bm{x})\leq \mathcal{L}w(t,\bm{x}),& \forall \bm{x}\in \mathcal{X}\setminus \mathcal{X}_s,  \forall t \in [0,T],\\
            v(t,\bm{x})\leq \frac{\partial w(t,\bm{x})}{\partial t},& \forall \bm{x}\in \partial \mathcal{X}, \forall t\in [0,T],
        \end{cases}
    \end{equation}
then, 
$\mathbb{P}_{\bm{x}_0}^{[0,T]}\geq v(0,\bm{x}_0)-\frac{2M}{T}$ for $\bm{x}_0\in \mathcal{X}\setminus \mathcal{X}_s$.    
\end{proposition}
\begin{proof}
From \eqref{upper_bound_1}, we have 
\[
        \begin{cases}
            \widehat{\mathcal{L}}v(t,\bm{x})\geq 0, & \forall \bm{x}\in \overline{\mathcal{X}\setminus \mathcal{X}_s},\forall t \in [0,T],\\
            v(t,\bm{x})\leq 1_{\partial \mathcal{X}_s}(\bm{x})+\widehat{\mathcal{L}}w(t,\bm{x}), & \forall \bm{x}\in \overline{\mathcal{X}\setminus \mathcal{X}_s}, \forall t \in [0,T]. 
        \end{cases}
\]

According to $\widehat{\mathcal{L}}v(t,\bm{x})\geq 0, \forall \bm{x}\in \overline{\mathcal{X}\setminus \mathcal{X}_s},\forall t \in [0,T]$, we have, for $t\in [0,T]$, that 
\[
\begin{aligned}
    \mathbb{E}[v(t,\widehat{\bm{X}}_{\bm{x}_0}^{\bm{w}}(t))]&=v(0,\bm{x}_0)+\mathbb{E}[\int_{0}^t \widehat{\mathcal{L}}v(\tau,\widehat{\bm{X}}_{\bm{x}_0}^{\bm{w}}(t))d\tau]\geq v(0,\bm{x}_0).
\end{aligned}
\]
Further, from  $v(t,\bm{x})\leq 1_{\partial \mathcal{X}_s}(\bm{x})+\widehat{\mathcal{L}}w(t,\bm{x}), \forall \bm{x}\in \overline{\mathcal{X}\setminus \mathcal{X}_s}, \forall t\in [0,T]$, we have 
\begin{equation*}
\begin{split}
& \hphantom{{}={}} \mathbb{P}(\widehat{\bm{X}}_{\bm{x}_0}^{\bm{w}}(T)\in \partial \mathcal{X}_s)=\mathbb{E}[1_{\partial \mathcal{X}_s}(\widehat{\bm{X}}_{\bm{x}_0}^{\bm{w}}(T))]\\
&\geq \frac{\int_{0}^T \mathbb{E}[1_{\partial \mathcal{X}_s}(\widehat{\bm{X}}_{\bm{x}_0}^{\bm{w}}(t))]dt}{T} (\text{according to Lemma \ref{equiv}})\\
&\geq \frac{\int_{0}^T \mathbb{E}[v(t,\widehat{\bm{X}}_{\bm{x}_0}^{\bm{w}}(t))]dt}{T}-\frac{\mathbb{E}[w(T,\widehat{\bm{X}}_{\bm{x}_0}^{\bm{w}}(T))]-w(0,\bm{x}_0)}{T}\\
&\geq v(0,\bm{x}_0)- \frac{\mathbb{E}[w(T,\widehat{\bm{X}}_{\bm{x}_0}^{\bm{w}}(T))]-w(0,\bm{x}_0)}{T}\\
&\geq v(0,\bm{x}_0)-\frac{2M}{T}.
  \end{split}
\end{equation*}
The proof is completed. 
\end{proof}

\begin{remark}
\label{Remark1}
The function $w(t,\bm{x})$ in Proposition \ref{coro4} cannot be removed. If this function is removed, we cannot find a  barrier function $v(t,\bm{x})$ satisfying condition \eqref{upper_bound_1} such that the lower bound $v(0,\bm{x}_0)$ of the reachability probability $\mathbb{P}_{\bm{x}_0}^{[0,T]}$ is larger than zero. A brief explanation is given here:   if the function $w(t,\bm{x})$ is removed, $v(t,\bm{x})$ will satisfy
\begin{equation*}
    \begin{cases}
        v(t,\bm{x})\leq 1, & \forall \bm{x}\in \partial \mathcal{X}_s, \forall t\in [0,T],\\
        v(t,\bm{x})\leq 0, & \forall \bm{x}\in \mathcal{X}\setminus \mathcal{X}_s, \forall t\in [0,T],
    \end{cases}
\end{equation*}
which implies $v(t,\bm{x})\leq 0, \forall \bm{x}\in \overline{\mathcal{X}\setminus \mathcal{X}_s}, \forall t\in [0,T]$, thus, $v(0,\bm{x}_0)$ will always be less than or equal to zero for $\bm{x}_0\in \mathcal{X}\setminus \mathcal{X}_s$. 
\end{remark}

Similar to Theorem \ref{coro2}, a barrier function that relaxes the submartingale requirement (i.e., $\mathcal{L}v(t,\bm{x})\geq 0, \forall \bm{x}\in \mathcal{X}\setminus \mathcal{X}_s,\forall t \in [0,T]$) in Proposition \ref{coro4} is formulated in Theorem \ref{coro5}.

\begin{theorem}
\label{coro5}
    Suppose there exist a barrier function $v(t,\bm{x})\colon [0,T]\times \mathbb{R}^n\rightarrow \mathbb{R}$ and a function $w(t,\bm{x})\colon [0,T]\times \mathbb{R}^n \rightarrow \mathbb{R}$ with $\sup_{(t,\bm{x})\in [0,T]\times \overline{\mathcal{X}}}|w(t,\bm{x})|\leq M$ that are continuously differentiable over $t$ and twice continuously differentiable over $\bm{x}$, satisfying 
   \begin{equation}
      \label{upper_bound_2}
        \begin{cases}
            \mathcal{L}v(t,\bm{x})\geq \alpha v(t,\bm{x})+\beta, & \forall \bm{x}\in \mathcal{X}\setminus \mathcal{X}_s,\forall t \in [0,T],\\
            \frac{\partial v(t,\bm{x})}{\partial t}\geq \alpha v(t,\bm{x})+\beta, & \forall \bm{x}\in \partial \mathcal{X}\cup \partial \mathcal{X}_s, \forall t\in [0,T],\\ 
            v(t,\bm{x})\leq 1+\frac{\partial w(t,\bm{x})}{\partial t}, &\forall \bm{x}\in \partial \mathcal{X}_s, \forall t\in [0,T],\\
            v(t,\bm{x})\leq \mathcal{L}w(t,\bm{x}), & \forall \bm{x}\in \mathcal{X}\setminus \mathcal{X}_s, \forall t \in [0,T],\\
            v(t,\bm{x})\leq \frac{\partial w(t,\bm{x})}{\partial t}, & \forall \bm{x}\in \partial \mathcal{X}, \forall t\in [0,T],
        \end{cases}
    \end{equation}
then
\begin{equation*}
\mathbb{P}_{\bm{x}_0}^{[0,T]}\geq 
\begin{cases}
\frac{(\frac{1}{\alpha}v(0,\bm{x}_0)+\frac{\beta}{\alpha^2})(e^{\alpha T}-1)-\frac{\beta}{\alpha}T}{T} -\frac{2M}{T}, & \text{if $\alpha\neq 0$}, \\
v(0,\bm{x}_0)+\frac{1}{2}\beta T-\frac{2M}{T}, & \text{if $\alpha=0$}
\end{cases}
\end{equation*}
for $\bm{x}_0\in \mathcal{X}\setminus \mathcal{X}_s$.    
\end{theorem}
\begin{proof}
From \eqref{upper_bound_2}, we have
 \begin{equation*}
        \begin{cases}
            \widehat{\mathcal{L}}v(t,\bm{x})\geq \alpha v(t,\bm{x})+\beta, & \forall \bm{x}\in \overline{\mathcal{X}\setminus \mathcal{X}_s},\forall t \in [0,T],\\
            v(t,\bm{x})\leq 1_{\partial \mathcal{X}_s}(\bm{x})+\widehat{\mathcal{L}}w(t,\bm{x}),&  \forall \bm{x}\in \overline{\mathcal{X}\setminus \mathcal{X}_s}, \forall t \in [0,T].
        \end{cases}
    \end{equation*}

When $\alpha\neq 0$, according to $\widehat{\mathcal{L}}v(t,\bm{x})\geq 0, \forall \bm{x}\in \overline{\mathcal{X}\setminus \mathcal{X}_s},\forall t \in [0,T]$, we have, for $t\in [0,T]$, that 
\begin{equation*}
\begin{split}
    \mathbb{E}[v(t,\widehat{\bm{X}}_{\bm{x}_0}^{\bm{w}}(t))]\geq e^{\alpha t}v(0,\bm{x}_0)+\frac{\beta}{\alpha}(e^{\alpha t}-1).
    \end{split}
    \end{equation*}

Further, from  $v(t,\bm{x})\leq 1_{\partial \mathcal{X}_s}(\bm{x})+\widehat{\mathcal{L}}w(t,\bm{x}), \forall t \in [0,T], \forall \bm{x}\in \overline{\mathcal{X}\setminus \mathcal{X}_s}$, we have 
\begin{equation*}
\begin{split}
&\hphantom{{}={}}\mathbb{P}(\widehat{\bm{X}}_{\bm{x}_0}^{\bm{w}}(T)\in \partial \mathcal{X}_s)=\mathbb{E}[1_{\partial \mathcal{X}_s}(\widehat{\bm{X}}_{\bm{x}_0}^{\bm{w}}(T))]\\
&\geq \frac{\int_{0}^T \mathbb{E}[1_{\partial \mathcal{X}_s}(\widehat{\bm{X}}_{\bm{x}_0}^{\bm{w}}(t))]dt }{T} (\text{according to Lemma \ref{equiv}})\\
&\geq \frac{\int_{0}^T \mathbb{E}[v(t,\widehat{\bm{X}}_{\bm{x}_0}^{\bm{w}}(t))]dt}{T}-\frac{\mathbb{E}[w(T,\widehat{\bm{X}}_{\bm{x}_0}^{\bm{w}}(T))]-w(0,\bm{x}_0)}{T}\\
&\geq \frac{\int_{0}^T e^{\alpha t}v(0,\bm{x}_0)+\frac{\beta}{\alpha}(e^{\alpha t}-1) dt}{T}- \frac{\mathbb{E}[w(T,\widehat{\bm{X}}_{\bm{x}_0}^{\bm{w}}(T))]-w(0,\bm{x}_0)}{T}\\
&\geq \frac{(\frac{1}{\alpha}e^{\alpha t}v(0,\bm{x}_0)+\frac{\beta}{\alpha^2}e^{\alpha t}-\frac{\beta}{\alpha}t)\mid_{0}^T}{T}-\frac{2M}{T}\\
&= \frac{(\frac{1}{\alpha}v(0,\bm{x}_0)+\frac{\beta}{\alpha^2})(e^{\alpha T}-1)-\frac{\beta}{\alpha}T}{T} -\frac{2M}{T}.
  \end{split}
\end{equation*}

The conclusion for $\alpha=0$ can be obtained via following the above procedure. The proof is completed. 
\end{proof}

Further, a straightforward result can be obtained from Theorem \ref{coro5} when searching for a time-independent function $v(\bm{x})\colon \mathbb{R}^n \rightarrow \mathbb{R}$ instead of a time-dependent function $v(t,\bm{x})\colon [0,T]\times \mathbb{R}^n \rightarrow \mathbb{R}$.

\begin{corollary}
\label{time-in-H}
    Suppose there exist twice continuously differentiable functions $v(\bm{x})\colon \mathbb{R}^n\rightarrow \mathbb{R}$ and $w(\bm{x})\colon \mathbb{R}^n \rightarrow \mathbb{R}$ with $\sup_{\bm{x}\in \overline{\mathcal{X}}}|w(\bm{x})|\leq M$, satisfying 
   \begin{equation}
      \label{upper_bound_3}
        \begin{cases}
            \mathcal{L}v(\bm{x})\geq \alpha v(\bm{x})+\beta, & \forall \bm{x}\in \mathcal{X}\setminus \mathcal{X}_s,\\
            0\geq \alpha v(\bm{x})+\beta, & \forall \bm{x}\in \partial \mathcal{X}\cup \partial \mathcal{X}_s\\ v(\bm{x})\leq 1 &\forall \bm{x}\in \partial \mathcal{X}_s,\\
            v(\bm{x})\leq \mathcal{L}w(\bm{x}), & \forall \bm{x}\in \mathcal{X}\setminus \mathcal{X}_s, \\
            v(\bm{x})\leq 0, & \forall \bm{x}\in \partial \mathcal{X},
        \end{cases}
    \end{equation}
then
\[\mathbb{P}_{\bm{x}_0}^{[0,T]}\geq
\begin{cases}
\frac{(\frac{1}{\alpha}v(\bm{x}_0)+\frac{\beta}{\alpha^2})(e^{\alpha T}-1)-\frac{\beta}{\alpha}T}{T} -\frac{2M}{T}, & \text{if $\alpha\neq 0$}, \\
v(\bm{x}_0)+\frac{1}{2}\beta T-\frac{2M}{T}, &\text{if $\alpha=0$}
\end{cases}
\]
for $\bm{x}_0\in \mathcal{X}\setminus \mathcal{X}_s$.    
\end{corollary}

\begin{remark}
When the subset $\mathcal{X}_s$ is equal to the complement of the state constrained set $\mathcal{X}$, i.e., $\mathcal{X}_s=\mathbb{R}^n\setminus  \mathcal{X}$, and $\mathcal{X}=\{\,\bm{x}\mid v(\bm{x})<1\,\}$ with $\partial \mathcal{X}=\{\,\bm{x}\mid v(\bm{x})=1\,\}$,  a condition that lower-bounds the reachability probability is formulated in Theorem 2 in  \cite{xue2023safe}, which is presented below,
   \begin{equation}
      \label{upper_bound_4}
        \begin{cases}
            \mathcal{L}v(\bm{x})\geq \alpha v(\bm{x})+\beta,~~~\forall \bm{x}\in \mathcal{X},\\
            \alpha+\beta>0.
        \end{cases}
    \end{equation}
Besides, when the $\mathcal{X}_s$ is a subset of the set $\mathcal{X}=\{\,\bm{x}\mid v(\bm{x})>0\,\}$ with $\partial \mathcal{X}=\{\,\bm{x}\mid v(\bm{x})=0\,\}$ and is equal to $\{\,\bm{x}\mid v(\bm{x})\geq 1\,\}$,  a condition that lower-bounds the reachability probability is formulated in Theorem 1 in  \cite{xue2023safe}, which is formulated below,
   \begin{equation}
      \label{upper_bound_5}
        \begin{cases}
            \mathcal{L}v(\bm{x})\geq \alpha v(\bm{x})+\beta, ~~~\forall \bm{x}\in \mathcal{X}\setminus \mathcal{X}_s,\\
            \alpha>-\beta\geq 0.
        \end{cases}
    \end{equation}

 When applied to the aforementioned cases in \cite{xue2023safe}, the condition \eqref{upper_bound_3} complements conditions \eqref{upper_bound_4} and \eqref{upper_bound_5} with $\alpha+\beta\leq 0$. 

\end{remark}

Besides, we can also obtain similar conclusions as in Remark \ref{re3} and \ref{re4}.

\section{Bounding Reachability Probabilities II}
\label{sec:bupII}
This section introduces our barrier functions for upper- and lower-bounding the reachability probability $\mathbb{P}_{\bm{x}_0}^{T}$ in Definition \ref{un_proII}.  They are respectively formulated in Subsection \ref{subsec:uupII} and \ref{subsec:lupII}.

Since the event for reachability at a specific time instant (Definition \ref{un_proII}) is fundamentally different from the one for a time horizon (Definition \ref{un_pro})—specifically, $A = \{\,\bm{w}\in \Omega \mid \bm{X}_{\bm{x}_0}^{\bm{w}}(T) \in \mathcal{X}_s\wedge \forall \tau \in [0,T). \bm{X}_{\bm{x}_0}^{\bm{w}}(\tau) \in \mathcal{X}\,\} \neq \{\,\bm{w}\in \Omega \mid \bm{X}_{\bm{x}_0}^{\bm{w}}(T) \in \partial \mathcal{X}_s\wedge \forall \tau \in [0,T). \bm{X}_{\bm{x}_0}^{\bm{w}}(\tau) \in \mathcal{X}\,\}= B$—the construction of the barrier functions in this section lies on a different auxiliary stochastic process $\{\widetilde{\bm{X}}_{\bm{x}_0}^{\bm{w}}(t), t\in \mathbb{R}_{\geq 0}\}$ for $\bm{x}_0\in \overline{\mathcal{X}\setminus \mathcal{X}_s}$ that is a stopped process corresponding to $\{\bm{X}_{\bm{x}_0}^{\bm{w}}(t), t\in [0,T^{\bm{x}_0}(\bm{w}))\}$ and the set $\overline{\mathcal{X}}$ rather than $\overline{\mathcal{X}\setminus \mathcal{X}_s}$ as $\{\widehat{\bm{X}}_{\bm{x}_0}^{\bm{w}}(t), t\in [0,T^{\bm{x}_0}(\bm{w}))\}$, i.e., 
\begin{equation}
\widetilde{\bm{X}}_{\bm{x}_0}^{\bm{w}}(t)=
\begin{cases}
\bm{X}_{\bm{x}_0}^{\bm{w}}(t), & \text{if $t<\tau^{\bm{x}_0}(\bm{w})$},\\
\bm{X}_{\bm{x}_0}^{\bm{w}}(\tau^{\bm{x}_0}(\bm{w})), & \text{if $t\geq \tau^{\bm{x}_0}(\bm{w})$},
\end{cases}
\end{equation}
where $\tau^{\bm{x}_0}(\bm{w})=\inf\{\,t\mid \bm{X}_{\bm{x}_0}^{\bm{w}}(t) \in \partial \mathcal{X}\,\}$ is the first time of exit of $\bm{X}_{\bm{x}_0}^{\bm{w}}(t)$ from the set $\mathcal{X}$. Similar to the stopped process in Section \ref{sec:bup}, the infinitesimal generator $\widetilde{\mathcal{L}}$ on a test function $v(t,\bm{x})$ corresponding to $\widetilde{\bm{X}}_{\bm{x}_0}^{\bm{w}}(t)$ is identical to the one corresponding to $\bm{X}_{\bm{x}_0}^{\bm{w}}(t)$ over $\mathcal{X}$, and is equal to $\frac{\partial v(t,\bm{x})}{\partial t}$ on $\partial \mathcal{X}$ \cite{kushner1967stochastic}. That is, for $v(t,\bm{x})$ which is twice continuously differentiable over $\bm{x}$ and continuously differentiable over $t$, 
\[
\begin{aligned}
\widetilde{\mathcal{L}}v(t,\bm{x})=\mathcal{L}v(t,\bm{x})=&\frac{\partial v(t,\bm{x})}{\partial t}+\frac{\partial v(t,\bm{x})}{\partial \bm{x}}\bm{b}(\bm{x})+\frac{1}{2}\textbf{\rm tr}(\bm{\sigma}(\bm{x})^{\top}\frac{\partial^2 v(t,\bm{x})}{\partial \bm{x}^2} \bm{\sigma}(\bm{x}))
\end{aligned}
\] for $(\bm{x},t)\in \mathcal{X} \times [0,T]$ and $\widetilde{\mathcal{L}}v(t,\bm{x})=\frac{\partial v(t,\bm{x})}{\partial t}$ for $\bm{x}\in \partial \mathcal{X}$ and $t\in [0,T]$.

The exact reachability probability $\mathbb{P}_{\bm{x}_0}^{T}$ is equal to the probability of reaching the set $\mathcal{X}_s$ at the time instant $T$ for the above auxiliary stochastic process. 

\begin{lemma}
\label{equivII}
    Given $\bm{x}_0\in \mathcal{X}\setminus \mathcal{X}_s$, \[
    \begin{split}
   & \mathbb{P}_{\bm{x}_0}^{T}=\mathbb{P}(\widetilde{\bm{X}}_{\bm{x}_0}^{\bm{w}}(T)\in \mathcal{X}_s)=\mathbb{E}[1_{\mathcal{X}_s}(\widetilde{\bm{X}}_{\bm{x}_0}^{\bm{w}}(T))].
 \end{split}
 \] 
\end{lemma}

The reachability probability $\mathbb{P}_{\bm{x}_0}^{T}$ can also be reduced to a solution to a second-order partial differential equation. 
\begin{proposition}
\label{eqII}
    Suppose there exists a function $v(t,\bm{x})\colon [0,T]\times \mathbb{R}^n$ which is twice continuously differentiable over $\bm{x}$ and continuously differentiable over $t$, satisfying 
    \begin{equation}
    \label{eq_proII}
        \begin{cases}
            \widetilde{\mathcal{L}}v(t,\bm{x})=0, & \forall \bm{x}\in \overline{\mathcal{X}},\forall t \in [0,T].\\
            v(T,\bm{x})=1_{\mathcal{X}_s}(\bm{x}), & \forall \bm{x}\in \overline{\mathcal{X}},
        \end{cases}
    \end{equation}
then
$\mathbb{P}_{\bm{x}_0}^{T}=\mathbb{P}(\widetilde{\bm{X}}_{\bm{x}_0}^{\bm{w}}(T)\in \mathcal{X}_s)=v(0,\bm{x}_0)$ for $\bm{x}_0\in \mathcal{X}\setminus \mathcal{X}_s$.    
\end{proposition}
\begin{proof}
The proof is similar to Proposition \ref{eq}. 
\end{proof}

\subsection{Upper-bounding Reachability Probabilities}
\label{subsec:uupII}
In this subsection, we present barrier functions for upper-bounding the reachability probability $\mathbb{P}_{\bm{x}_0}^{T}$ in Definition \ref{un_proII}. The first time-dependent barrier function is obtained via relaxing equation \eqref{eq_proII} in Proposition \ref{eqII}. The second one extends upon the first one, which relaxes the supermartingale requirement. The third one is a variant of the second one, using a time-independent function $v(\bm{x})$ instead of a time-dependent one $v(t,\bm{x})$. They are respectively formulated in Corollary \ref{coro1II},  Theorem \ref{coro2II}, and Corollary \ref{coro3II}.

\begin{corollary}
\label{coro1II}
    Suppose there exists a barrier function $v(t,\bm{x})\colon [0,T]\times \mathbb{R}^n\rightarrow \mathbb{R}$ that is continuously differentiable over $t$ and twice continuously differentiable over $\bm{x}$, satisfying 
   \begin{equation}
   \label{lower_bound_condition1II}
        \begin{cases}
            \mathcal{L}v(t,\bm{x})\leq 0, & \forall \bm{x}\in \mathcal{X},\forall t \in [0,T],\\
            \frac{\partial v(t,\bm{x})}{\partial t}\leq 0, & \forall \bm{x}\in \partial \mathcal{X}, \forall t\in [0,T],\\ 
            v(T,\bm{x})\geq 1_{\mathcal{X}_s}(\bm{x}), & \forall \bm{x}\in \overline{\mathcal{X}},
        \end{cases}
    \end{equation}
then
$\mathbb{P}_{\bm{x}_0}^{T}\leq v(0,\bm{x}_0)$ for $\bm{x}_0\in \mathcal{X}\setminus \mathcal{X}_s$.    
\end{corollary}
\begin{proof}
From \eqref{lower_bound_condition1II}, we have 
 \begin{equation}
 \label{lower_bound_condition11II}
        \begin{cases}
            \widetilde{\mathcal{L}}v(t,\bm{x})\leq 0, & \forall \bm{x}\in \overline{\mathcal{X}},\forall t \in [0,T],\\
            v(T,\bm{x})\geq 1_{\mathcal{X}_s}(\bm{x}), & \forall \bm{x}\in \overline{\mathcal{X}}.
        \end{cases}
    \end{equation}

Then, we can obtain the conclusion via using Lemma \ref{equivII} and following the proof of Corollary \ref{coro1}.  
\end{proof}

Corollary \ref{coro1II} states that if there exists a barrier function $v(t,\bm{x})\colon [0,T]\times \mathbb{R}^n\rightarrow \mathbb{R}$ satisfying \eqref{lower_bound_condition1II}, then the reachability probability $\mathbb{P}_{\bm{x}_0}^{T}$ can be bounded above by $v(0,\bm{x}_0)$. Analogously, the requirement for $\mathcal{L}v(t,\bm{x})\leq 0$ to hold for $(t,\bm{x}) \in [0,T]\times \mathcal{X}$ may hinder the acquisition of such a barrier function. We will relax this requirement below.

\begin{theorem}
\label{coro2II}
    Suppose there exists a barrier function $v(t,\bm{x})\colon [0,T]\times \mathbb{R}^n\rightarrow \mathbb{R}$ that is continuously differentiable over $t$ and twice continuously differentiable over $\bm{x}$, satisfying 
 \begin{equation}
 \label{lower_bound_condition2II}
     \begin{cases}
        \mathcal{L}v(t,\bm{x})\leq \alpha v(t,\bm{x})+\beta, &\forall \bm{x}\in \mathcal{X}, \forall t\in [0,T],\\
        \frac{\partial v(t,\bm{x})}{\partial t} \leq \alpha v(t,\bm{x})+\beta, & \forall \bm{x}\in \partial \mathcal{X}, \forall t\in [0,T],\\
        v(T,\bm{x})\geq 1_{\mathcal{X}_s}(\bm{x}), & \forall \bm{x}\in \overline{\mathcal{X}},
    \end{cases}
\end{equation}
then
\begin{equation*}
\mathbb{P}_{\bm{x}_0}^{T}\leq
    \begin{cases}
      v(0,\bm{x}_0)+\beta T, & \text{if $\alpha=0$},\\
      e^{\alpha T}v(0,\bm{x}_0)+\frac{\beta}{\alpha}(e^{\alpha T}-1), & \text{if  $\alpha\neq 0$}
    \end{cases}
\end{equation*}
 for $\bm{x}_0\in \mathcal{X}\setminus \mathcal{X}_s$.    
\end{theorem}
\begin{proof}
From \eqref{lower_bound_condition2II}, we have 
 \begin{equation*}
        \begin{cases}
            \widetilde{\mathcal{L}}v(t,\bm{x})\leq \alpha v(t,\bm{x})+\beta, & \forall \bm{x}\in \overline{\mathcal{X}}, \forall t \in [0,T],\\
            v(T,\bm{x})\geq 1_{\mathcal{X}_s}(\bm{x}),& \forall \bm{x}\in \overline{\mathcal{X}}.
        \end{cases}
    \end{equation*}

Then, we can obtain the conclusion via using Lemma \ref{equivII} and following the proof of Corollary \ref{coro2}.  
\end{proof}

A straightforward result is obtained from Theorem \ref{coro2II} when searching for a time-independent barrier function $v(\bm{x})\colon \mathbb{R}^n \rightarrow \mathbb{R}$ instead of a time-dependent one $v(t,\bm{x})\colon [0,T]\times \mathbb{R}^n \rightarrow \mathbb{R}$. 
\begin{corollary}
\label{coro3II}
    Suppose there exists a barrier function $v(\bm{x})\colon \mathbb{R}^n\rightarrow \mathbb{R}$, which is twice continuously differentiable, satisfying 
      \begin{equation}
    \label{lower_bound_3II}
        \begin{cases}
            \mathcal{L}v(\bm{x})\leq \alpha v(\bm{x})+\beta, & \forall \bm{x}\in \mathcal{X}, \\
            0\leq \alpha v(\bm{x})+\beta,& \forall \bm{x}\in \partial \mathcal{X},\\
            v(\bm{x})\geq 1_{\mathcal{X}_s}(\bm{x}),& \forall \bm{x}\in \overline{\mathcal{X}},
        \end{cases}
    \end{equation} 
then
\begin{equation*}
\mathbb{P}_{\bm{x}_0}^{T}\leq
    \begin{cases}
      v(\bm{x}_0)+\beta T, & \text{if $\alpha=0$},\\
      e^{\alpha T}v(\bm{x}_0)+\frac{\beta}{\alpha}(e^{\alpha T}-1), &\text{if  $\alpha\neq 0$}
    \end{cases}
\end{equation*}
 for $\bm{x}_0\in \mathcal{X}\setminus \mathcal{X}_s$.    
\end{corollary}

\subsection{Lower-bounding Reachability Probabilities}
\label{subsec:lupII}
In this subsection, we present our barrier functions for lower-bounding the reachability probability $\mathbb{P}_{\bm{x}_0}^{T}$, based on the following assumption  on $\mathcal{X}_s$. 

\begin{assumption}
\label{inter-assum}
The set $\mathcal{X}_s$ has non-empty interior, i.e., $\mathcal{X}_s^\circ \neq \emptyset$.
\end{assumption}

This assumption is necessary to ensure the possibility of a non-trivial (positive) lower bound on $\mathbb{P}_{\boldsymbol{x}_0}^{T}$ in the following constructed barrier functions. The conditions for lower-bounding $\mathbb{P}_{\boldsymbol{x}_0}^{T}$ require a barrier function $v(t, \boldsymbol{x})$ to satisfy $v(T, \boldsymbol{x}) \leq \mathbf{1}_{\mathcal{X}_s}(\boldsymbol{x})$ for all $\boldsymbol{x} \in \overline{\mathcal{X}}$.
\begin{itemize}
    \item If $\mathcal{X}_s$ has no interior (e.g., it is a single point or a lower-dimensional manifold), then $1_{\mathcal{X}_s}(\boldsymbol{x}) = 0$ for almost every $\boldsymbol{x} \in \overline{\mathcal{X}}$.
    \item Given the continuity of $v(T, \boldsymbol{x})$, the condition would then simplify to $v(T, \boldsymbol{x}) \leq 0$ everywhere.
    \item This, in combination with the generator conditions, would typically force the obtained lower bounds to be non-positive.
\end{itemize}

Like the one in Corollary \ref{coro1II}, the construction of the first barrier function was obtained by relaxing equation \eqref{eq_proII}.

\begin{corollary}
\label{coro4II}
    Suppose there exists a barrier function $v(t,\bm{x})\colon [0,T]\times \mathbb{R}^n\rightarrow \mathbb{R}$ that is continuously differentiable over $t$ and twice continuously differentiable over $\bm{x}$, satisfying 
   \begin{equation}
   \label{upper_bound_1II}
        \begin{cases}
            \mathcal{L}v(t,\bm{x})\geq 0, & \forall \bm{x}\in \mathcal{X},\forall t \in [0,T],\\
            \frac{\partial v(t,\bm{x})}{\partial t}\geq 0, & \forall \bm{x}\in \partial \mathcal{X}, \forall t\in [0,T],\\ 
            v(T,\bm{x})\leq 1_{\mathcal{X}_s}(\bm{x}), & \forall \bm{x}\in \overline{\mathcal{X}},
        \end{cases}
    \end{equation}
then
$\mathbb{P}_{\bm{x}_0}^{T}\geq v(0,\bm{x}_0)$ for $\bm{x}_0\in \mathcal{X}\setminus \mathcal{X}_s$.    
\end{corollary}
\begin{proof}
From \eqref{upper_bound_1II}, we have 
 \begin{equation*}
        \begin{cases}
            \widetilde{\mathcal{L}}v(t,\bm{x})\geq 0, & \forall \bm{x}\in \overline{\mathcal{X}},\forall t \in [0,T],\\
            v(T,\bm{x})\leq 1_{\mathcal{X}_s}(\bm{x}), & \forall \bm{x}\in \overline{\mathcal{X}}.
        \end{cases}
    \end{equation*}
    
Then, we can obtain the conclusion via using Lemma \ref{equivII} and following the proof of Corollary \ref{coro1}.  
\end{proof}

Based on the stochastic process $\widetilde{\bm{X}}_{\bm{x}_0}^{\bm{w}}(\cdot)$, one might wander the possibility of constructing a barrier function, similar to the one in Corollary \ref{coro4II}, to lower-bound $\mathbb{P}_{\bm{x}_0}^{T}$. Let us consider the existence of a barrier function $v(t,\bm{x})\colon [0,T]\times \mathbb{R}^n\rightarrow \mathbb{R}$ that is twice continuously differentiable with respect to $\bm{x}$ and continuously differentiable over $t$, satisfying 
   \begin{equation}
   \label{upper_bound_1III}
        \begin{cases}
            \mathcal{L}v(t,\bm{x})\geq 0, & \forall \bm{x}\in \mathcal{X},\forall t \in [0,T],\\
            \frac{\partial v(t,\bm{x})}{\partial t}\geq 0, & \forall \bm{x}\in \mathcal{X}_s,\forall t \in [0,T],\\
            \frac{\partial v(t,\bm{x})}{\partial t}\geq 0, & \forall \bm{x}\in \partial \mathcal{X}, \forall t\in [0,T],\\ 
            v(T,\bm{x})\leq 1_{\mathcal{X}_s}(\bm{x}), & \forall \bm{x}\in \overline{\mathcal{X}}.
        \end{cases}
    \end{equation}
 We can conclude that $\mathbb{P}_{\bm{x}_0}^{[0,T]}\geq v(0,\bm{x}_0)$ for $\bm{x}_0\in \mathcal{X}\setminus \mathcal{X}_s$. However, it is observed that if $v(\bm{x})$ satisfies \eqref{upper_bound_1III}, it also satisfies \eqref{upper_bound_1} with $w(\bm{x})\equiv 0$.  As commented in Remark \ref{Remark1}, we cannot obtain meaningful results.

Similar to Theorem \ref{coro2II}, a condition that relaxes the submartingale requirement (i.e., $\mathcal{L}v(t,\bm{x})\geq 0, \forall \bm{x}\in \mathcal{X},\forall t \in [0,T]$) in Proposition \ref{coro4II} is formulated in Theorem \ref{coro5II}.

\begin{theorem}
\label{coro5II}
    Suppose there exist a barrier function $v(t,\bm{x})\colon [0,T]\times \mathbb{R}^n\rightarrow \mathbb{R}$ that is continuously differentiable over $t$ and twice continuously differentiable over $\bm{x}$, satisfying 
   \begin{equation}
      \label{upper_bound_2II}
        \begin{cases}
            \mathcal{L}v(t,\bm{x})\geq \alpha v(t,\bm{x})+\beta, &\forall \bm{x}\in \mathcal{X},\forall t \in [0,T],\\
            \frac{\partial v(t,\bm{x})}{\partial t}\geq \alpha v(t,\bm{x})+\beta,& \forall \bm{x}\in \partial \mathcal{X}, \forall t\in [0,T],\\ 
            v(T,\bm{x})\leq 1_{\mathcal{X}_s}(\bm{x}),& \forall \bm{x}\in \overline{\mathcal{X}},
        \end{cases}
    \end{equation}
then
\begin{equation*}
\mathbb{P}_{\bm{x}_0}^{T}\geq 
\begin{cases}
 e^{\alpha T}v(0,\bm{x}_0)+\frac{\beta}{\alpha}(e^{\alpha T}-1), &\text{if $\alpha\neq 0$},\\
v(0,\bm{x}_0)+\beta T, & \text{if $\alpha=0$}
\end{cases}
\end{equation*}
for $\bm{x}_0\in \mathcal{X}\setminus \mathcal{X}_s$.    
\end{theorem}
\begin{proof}
From \eqref{upper_bound_2II}, we have
 \begin{equation*}
        \begin{cases}
            \widetilde{\mathcal{L}}v(t,\bm{x})\geq \alpha v(t,\bm{x})+\beta, &\forall \bm{x}\in \overline{\mathcal{X}},\forall t \in [0,T],\\
            v(T,\bm{x})\leq 1_{\partial \mathcal{X}_s}(\bm{x}), & \forall \bm{x}\in \overline{\mathcal{X}}.
        \end{cases}
    \end{equation*}

If $\alpha\neq 0$, according to $\widetilde{\mathcal{L}}v(t,\bm{x})\geq \alpha v(t,\bm{x})+\beta, \forall \bm{x}\in \overline{\mathcal{X}},\forall t \in [0,T]$, we have, for $t\in [0,T]$, that 
\begin{equation*}
\begin{split}
    \mathbb{E}[v(t,\widetilde{\bm{X}}_{\bm{x}_0}^{\bm{w}}(t))]\geq e^{\alpha t}v(0,\bm{x}_0)+\frac{\beta}{\alpha}(e^{\alpha t}-1).
    \end{split}
    \end{equation*}

Further, from  $v(T,\bm{x})\leq 1_{\partial \mathcal{X}_s}(\bm{x}), \forall \bm{x}\in \overline{\mathcal{X}}$, we have 
\begin{equation*}
\begin{split}
\mathbb{P}(\widetilde{\bm{X}}_{\bm{x}_0}^{\bm{w}}(T)\in \mathcal{X}_s)&=\mathbb{E}[1_{\mathcal{X}_s}(\widetilde{\bm{X}}_{\bm{x}_0}^{\bm{w}}(T))]\geq  \mathbb{E}[v(T,\widetilde{\bm{X}}_{\bm{x}_0}^{\bm{w}}(T))]\\
&\geq e^{\alpha T}v(0,\bm{x}_0)+\frac{\beta}{\alpha}(e^{\alpha T}-1).
  \end{split}
\end{equation*}

If $\alpha=0$, we can obtain $\mathbb{P}_{\bm{x}_0}^{T}\geq v(0,\bm{x}_0)+\beta T$ by following the above procedure. 

The proof is completed. 
\end{proof}

\begin{remark}
    Theorem \ref{coro5II} provides a time-dependent barrier function for lower-bounding $\mathbb{P}_{\boldsymbol{x}{0}}^{T}$. A natural follow-up step would be to seek a time-independent version of this result, analogous to Corollaries \ref{coro3} and \ref{coro3II}. Such a formulation 
    \begin{equation}
      \label{upper_bound_3II}
        \begin{cases}
            \mathcal{L}v(\bm{x})\geq \alpha v(\bm{x})+\beta, &\forall \bm{x}\in \mathcal{X},\\
            0\geq \alpha v(\bm{x})+\beta, & \forall \bm{x}\in \partial \mathcal{X},\\ 
            v(\bm{x})\leq 1_{\mathcal{X}_s}(\bm{x}), & \forall \bm{x}\in \overline{\mathcal{X}},
        \end{cases}
    \end{equation}
leads to 
\[
\mathbb{P}_{\bm{x}_0}^{T}\geq 
\begin{cases}
e^{\alpha T}v(\bm{x}_0)+\frac{\beta}{\alpha}(e^{\alpha T}-1), & \text{if $\alpha\neq 0$},\\
v(\bm{x}_0)+\beta T, & \text{if $\alpha=0$}
\end{cases}
\]
for $\bm{x}_0\in \mathcal{X}\setminus \mathcal{X}_s$. However, these bounds will be  trivial (non-positive). For the detailed explanation, please refer to Appendix \ref{section:expla}. Therefore, a time-independent barrier function is not a viable approach for this specific reachability problem. 
\end{remark}

\section{Examples}
\label{sec:ex}
In this section, we assess the effectiveness of the proposed barrier certificates in bounding reachability probabilities through numerical examples. As a benchmark, we compute empirical estimates of the exact reachability probabilities using $10^4$ Monte Carlo simulations with the Euler--Maruyama method. All conditions \eqref{fen2020}, \eqref{lower_bound_condition2},  \eqref{lower_bound_3},\eqref{lower_bound_4}, \eqref{upper_bound_2}, \eqref{upper_bound_3}, \eqref{lower_bound_condition2II}, \eqref{lower_bound_3II}, and \eqref{upper_bound_2II} are encoded into semidefinite programs via SOS decomposition \cite{papachristodoulou2005tutorial}, solved with MOSEK 10.2 \cite{aps2019mosek}. To avoid the bilinearity that would arise from jointly optimizing over both the parameter $\alpha$ and the barrier function $v$, we adopt a pragmatic strategy: we carry out a simple one-dimensional exploration by manually trying different values of $\alpha$. For each fixed $\alpha$, the problem of finding a suitable barrier function $v(t,\bm{x})$ or $v(\bm{x})$ becomes a convex optimization problem. Unless otherwise stated, we utilize polynomial barrier functions of degree $d$, meaning they include all monomials with total degree less than or equal to $d$.

\begin{example}[Population Growth Model]
\label{ex1}
Consider the following system:
\[
dX(t,w)=b(X(t,w))dt+\sigma(X(t,w))dW(t,w),
\]
with state constraint $\mathcal{X}=\{x\in\mathbb{R}\mid x^2-1<0\}$ and safe set $\mathcal{X}_s=\{x\in\mathbb{R}\mid 100x^2-1\leq 0\}$.

\textbf{Case 1.} $b(x)=-x$, $\sigma(x)=\tfrac{\sqrt{2}}{2}x$, $T=100$, $x_0=-0.8$. Fig.~\ref{fig:ex11_case1} illustrates trajectories in the $x$--$t$ space. Computed bounds are summarized in Table~\ref{tab:ex1_case1}.  

\begin{figure}[h]
    \centering
    \includegraphics[width=0.5\linewidth]{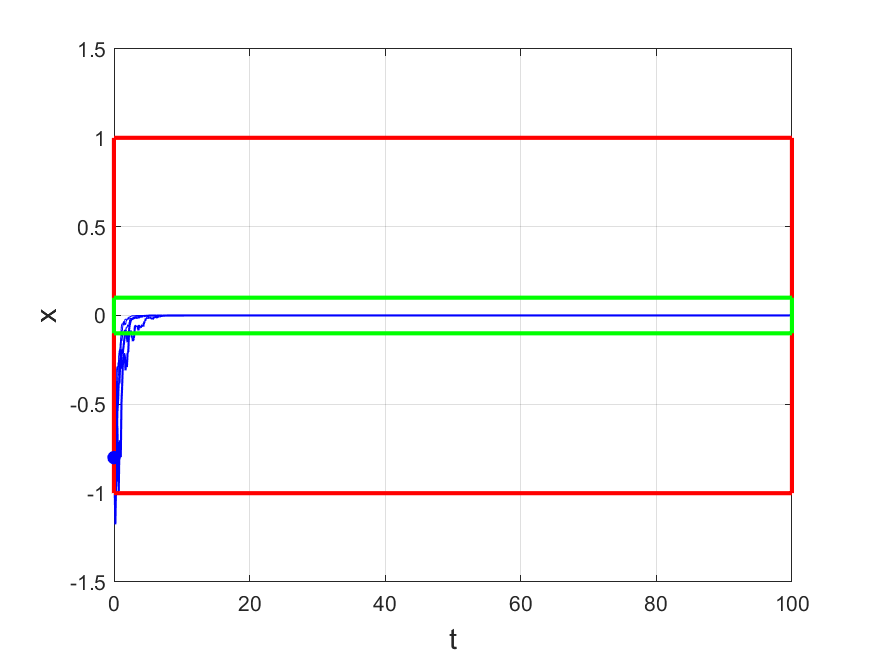}
    \caption{Case 1 of Example~1. 
    The region bounded by the red curve represents the safe set over the horizon $[0,100]$, 
    the region bounded by the green curve represents the target set over the horizon $[0,100]$, 
    and the blue curves represent five simulated trajectories starting from the initial state $x_0 = -0.8$.}
    \label{fig:ex11_case1}
\end{figure}

\begin{table}[H]
\centering
\caption{Bounds for Example~\ref{ex1}, Case 1 ($b(x)=-x$, $\sigma(x)=\tfrac{\sqrt{2}}{2}x$, $T=100$, $x_0=-0.8$).}
\label{tab:ex1_case1}
\begin{tabular}{lccc}
\hline
Condition & $\alpha$ & Degree & Bound \\
\hline
\eqref{lower_bound_condition2} & 0 &  6 in $x$, 1 in $t$ & 0.6723\\
\eqref{lower_bound_3} & $0$ & 6 & $0.6723$ \\
\eqref{upper_bound_3} & $0$ & 14 & $0.6423$ \\
\eqref{lower_bound_condition2II} & $0$ & 6 in $x$, 1 in $t$ & $0.6723$ \\
\eqref{upper_bound_2II} & $0$ & 14 in $x$, 1 in $t$ & $0.6480$ \\
\hline
Monte Carlo & --- & --- & $\mathbb{P}_{x_0}^{[0,T]}=0.6556$ \\
Monte Carlo & --- & --- & $\mathbb{P}_{x_0}^{T}=0.6556$ \\
\hline
\end{tabular}
\end{table}

\emph{Remark.} The bounds obtained from time-independent conditions are already tight; time-dependent barrier functions give comparable results but sometimes suffer from numerical instability.

\textbf{Case 2.} $b(x)=-10x$, $\sigma(x)=\tfrac{\sqrt{2}}{2}x$, $T=1$, $x_0=-0.8$.  Fig.~\ref{fig:ex11_case2} illustrates trajectories in the $x$--$t$ space. Table~\ref{tab:ex1_case2} shows the computed bounds.

\begin{figure}[h]
    \centering
    \includegraphics[width=0.5\linewidth]{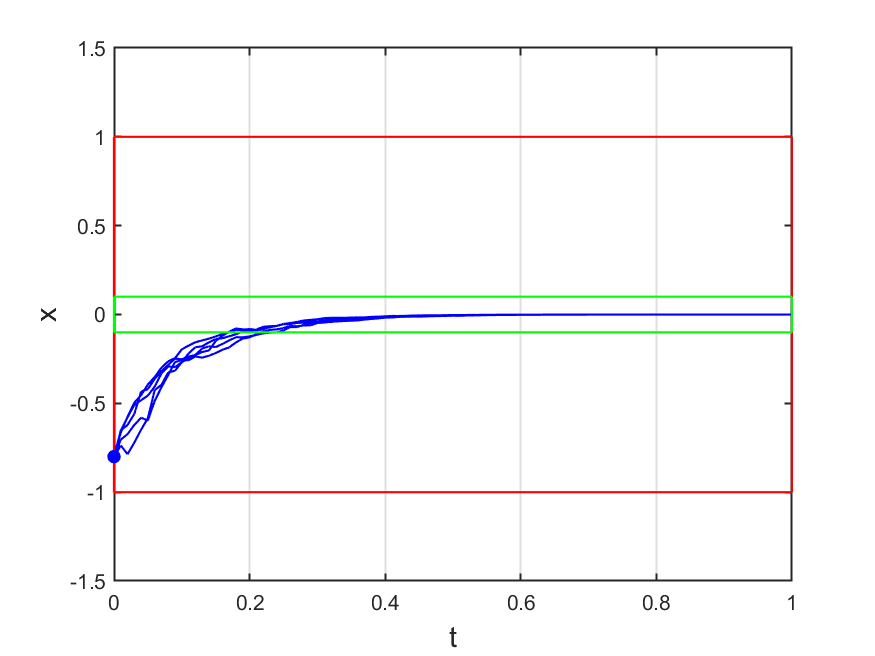}
    \caption{Case 2 of Example~1. 
    The region bounded by the red curve represents the safe set over the horizon $[0,1]$, 
    the region bounded by the green curve represents the target set over the horizon $[0,1]$, 
    and the blue curves represent five simulated trajectories starting from the initial state $x_0 = -0.8$.}
    \label{fig:ex11_case2}
\end{figure}

\begin{table}[H]
\centering
\caption{Bounds for Example~\ref{ex1}, Case 2 ($b(x)=-10x$, $\sigma(x)=\tfrac{\sqrt{2}}{2}x$, $T=1$, $x_0=-0.8$).}
\label{tab:ex1_case2}
\begin{tabular}{lccc}
\hline
Condition & $\alpha$ & Degree & Bound \\
\hline
\eqref{upper_bound_2} & $0$ & 8 & $0.7253$ \\
\eqref{upper_bound_3} & $0$ & 8 & $0.4912$ \\
\eqref{upper_bound_2II} & $0$ & 8 & $0.9052$ \\
\eqref{upper_bound_2} & $1$ & 8 & $0.7218$ \\
\eqref{upper_bound_3} & $1$ & 8 & $0.4774$ \\
\eqref{upper_bound_2II} & $1$ & 8 & $0.9349$ \\
\hline
Monte Carlo & --- & --- & $\mathbb{P}_{x_0}^{[0,T]}=1.0000$ \\
Monte Carlo & --- & --- & $\mathbb{P}_{x_0}^T=1.0000$\\
\hline
\end{tabular}
\end{table}

\emph{Remark.} 
While the time-dependent conditions \eqref{upper_bound_2} and \eqref{upper_bound_3} provided valid lower bounds (\(0.7253\) and \(0.4912\) with \(\alpha=0\)), even tighter lower bounds for \(\mathbb{P}_{x_0}^{[0,T]}\) were obtained by using the barrier function from \eqref{upper_bound_2II} designed for \(\mathbb{P}_{x_0}^{T}\), yielding a bound of \(0.9052\) (\(\alpha=0\)) and \(0.9349\) (\(\alpha=1.0\)). This conclusion applies to the case with $\alpha=1$ as well. This illustrates the potential for cross-application of the barrier certificates between the two reachability problems to achieve sharper results.

\textbf{Case 3.} $b(x)=-x+0.1$, $\sigma(x)=x^2$, $T=10$, $x_0=-0.5$.  
 Fig.~\ref{fig:ex11_case3} illustrates trajectories in the $x$--$t$ space.
Table~\ref{tab:ex1_case3} shows the computed bounds.  

\begin{figure}[h]
    \centering
    \includegraphics[width=0.5\linewidth]{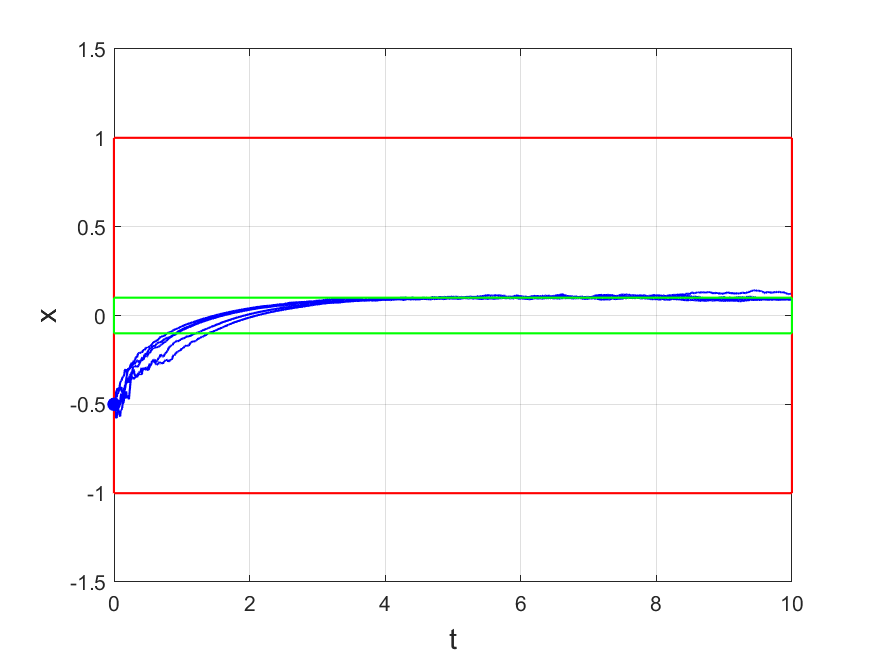}
    \caption{Case 3 of Example~1. 
    The region bounded by the red curve represents the safe set over the horizon $[0,10]$, 
    the region bounded by the green curve represents the target set over the horizon $[0,10]$, 
    and the blue curves represent five simulated trajectories starting from the initial state $x_0 = -0.5$.}
    \label{fig:ex11_case3}
\end{figure}

\begin{table}[H]
\centering
\caption{Bounds for Example~\ref{ex1}, Case 3 ($b(x)=-x+0.1$, $\sigma(x)=x^2$, $T=10$, $x_0=-0.5$).}
\label{tab:ex1_case3}
\begin{tabular}{lccc}
\hline
Condition & $\alpha$ & Degree & Bound \\
\hline
\eqref{upper_bound_2} & $0$ & 8 in $x$, 1 in $t$ & $0.6830$ \\
\eqref{upper_bound_3} & $0$ & 8 & $0.6723$ \\
\eqref{upper_bound_3} & $0$ & 14 & $0.7649$ \\
\eqref{upper_bound_2} & $0.01$ & 8 in $x$, 1 in $t$ & $0.7173$ \\
\eqref{upper_bound_3} & $0.01$ & 8 & $0.7074$ \\
\eqref{upper_bound_3} & $0.01$ & 14 & $0.8035$ \\
\hline
Monte Carlo & --- & --- & $\mathbb{P}_{x_0}^{[0,T]}=0.9842$ \\
Monte Carlo & --- & --- & $\mathbb{P}_{x_0}^T=0.4623$ \\
\hline
\end{tabular}
\end{table}

\emph{Remark.} For this nonlinear system, increasing the degree and adjusting $\alpha$ improves tightness, though numerical issues appear for higher-degree time-dependent barriers.
\end{example}

\begin{example}
\label{ex2}

Consider the following two-dimensional system:
\begin{equation*}
\begin{cases}
dx(t,w)=(-x(t,w)+1)dt+x(t,w)^2dW_1(t,w),\\
dy(t,w)=(10y(t,w)+x(t,w))dt-x(t,w)dW_2(t,w),
\end{cases}
\end{equation*}
with $\mathcal{X}=\{x^2+y^2\leq 1\}$, $\mathcal{X}_s=\{x^2+y^2\leq 0.01\}$, $x_0=(-0.5,0.5)$, $T=1$.  


Fig.~\ref{fig:ex2} illustrates trajectories in the $x$--$y$ space. Computed bounds are shown in Table~\ref{tab:ex2}.  

\begin{figure}[h]
    \centering
    \includegraphics[width=0.4\linewidth]{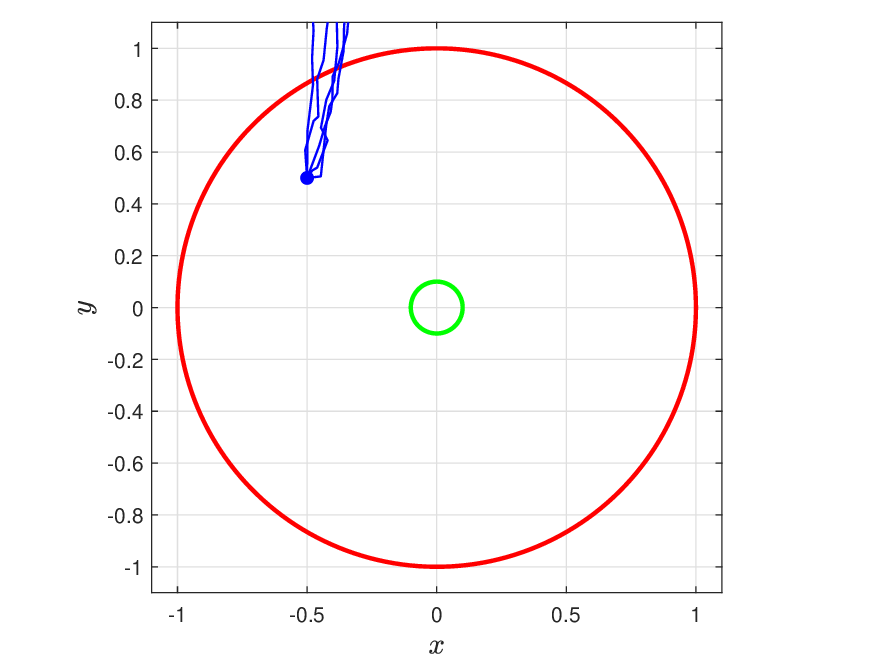}
    \caption{ Example~2. 
    The region bounded by the red curve represents the safe set, 
    the region bounded by the green curve represents the target set, 
    and the blue curves represent five simulated trajectories starting from the initial state $(x_0,y_0)^{\top}=(-0.5,0.5)^{\top}$.}
    \label{fig:ex2}
\end{figure}

\begin{table}[H]
\centering
\caption{Bounds for Example~\ref{ex2}.}
\label{tab:ex2}
\begin{tabular}{lccc}
\hline
Condition & $\alpha$ & Degree & Bound \\
\hline
\eqref{fen2020} & --- & 4 & $0.3317$ \\
\eqref{lower_bound_condition2} & $0$ & 4 & $0.3365$ \\
\eqref{lower_bound_3} & $0$ & 4 & $0.3384$ \\
\eqref{lower_bound_4} & $0$ & 4 & $0.3384$ \\
\eqref{lower_bound_condition2II} & $0$ & 4 & $0.1465$ \\
\eqref{lower_bound_3II} & $0$ & 4 & $0.3384$ \\
\eqref{lower_bound_condition2} & $3.0$ & 4 & $0.2292$ \\
\eqref{lower_bound_3} & $3.0$ & 4 & $1.0000$ \\
\eqref{lower_bound_condition2II} & $2.0$ & 4 & $0.0670$ \\
\eqref{lower_bound_3II} & $2.0$ & 4 & $0.9735$ \\
\hline
Monte Carlo & --- & --- & $\mathbb{P}_{x_0}^{[0,T]}=0.0000$ \\
Monte Carlo & --- & --- & $\mathbb{P}_{x_0}^{T}=0.0000$ \\
\hline
\end{tabular}
\end{table}

\emph{Remark.} 
Using time-dependent barrier functions (condition~\eqref{lower_bound_condition2}, particularly with $\alpha = 3.0$) yields tighter upper bounds than the time-independent counterpart (condition~\eqref{lower_bound_3}) or existing methods (conditions~\eqref{fen2020} and \eqref{lower_bound_4} with $\alpha = 0$). For condition~\eqref{lower_bound_4} with $\alpha<0$, we did not obtain any meaningful upper bounds. In addition, the barrier in~\eqref{lower_bound_condition2II} with $\alpha=2.0$, designed for the specific time-instant problem, provide the tightest upper bound $0.0670$ for $\mathbb{P}_{\boldsymbol{x}_0}^{T}$, which is close to the empirical estimate $0$ obtained from the Monte Carlo method.

\end{example}


The above examples demonstrate that:
\begin{itemize}
   \item \textbf{Comparative Advantage:} The new conditions often provide tighter bounds (e.g., Example 2) and more flexible stability ranges for parameters (e.g., \(\alpha \neq 0\)) compared to previous methods, e.g., \eqref{fen2020} and \eqref{lower_bound_4}.
    \item \textbf{Cross-Application:} Barrier functions designed for one type of reachability probability (e.g., \(\mathbb{P}_{\boldsymbol{x}_0}^{T}\)) can sometimes be effectively used to bound the other type (e.g., \(\mathbb{P}_{\boldsymbol{x}_0}^{[0,T]}\)), providing a useful strategy for obtaining tighter results.
    \item \textbf{Parameter Tuning:} The parameter \(\alpha\) in the conditions provides a degree of freedom that can be tuned to optimize the tightness of the bounds for a given system.
    \item \textbf{Computational Trade-offs:} Time-dependent barrier functions can yield sharper results but introduce an additional variable $t$, which increases the SDP size. This increased dimensionality often causes numerical instability for higher polynomial degrees, suggesting a practical trade-off between bound tightness and computational tractability.
\end{itemize}

\section{Conclusion}
\label{sec:con}
This paper introduced a  new framework of constructing  barrier functions to establish both upper and lower bounds on the probabilities of reaching specific sets over finite time horizons and at finite time instants in continuous-time stochastic systems described by SDEs. These proposed barrier functions offer stronger alternatives, complement existing methods, or fill gaps, facilitating the calculation of precise bounds on reachability probabilities.

In future work, we will develop advanced numerical methods for computing the proposed barrier functions. Moreover, this framework provides a unifying foundation for reachability analysis, as the stochastic barrier conditions naturally reduce to their deterministic counterparts when the diffusion term vanishes (see Appendix~\ref{app:deter}). A more detailed comparison with specialized deterministic methods will also be pursued.
\section*{Acknowledgements}
This work is funded by the CAS Pioneer Hundred Talents Program, Basic Research Program of  Institute of Software, CAS (Grant No. ISCAS-JCMS-202302), and NRF RSS Scheme NRF-RSS2022-009.

\bibliographystyle{plain}         
\bibliography{autosam}       

\section*{Appendix}

\subsection{Semi-definite Programming Implementation}

In this section we present the semi-definite programs to optimize the lower and upper bounds of the reachability probabilities $\mathbb{P}_{\bm{x}_0}^{[0,T]}$ and $\mathbb{P}_{\bm{x}_0}^{T}$. These semi-definite programs are constructed via encoding constraints 
\eqref{fen2020}, \eqref{lower_bound_condition2}, \eqref{lower_bound_3}, \eqref{upper_bound_2}, \eqref{upper_bound_3}, \eqref{lower_bound_condition2II}, \eqref{coro3II}, and \eqref{upper_bound_3II}  into semi-definite constraints using sum of squares decomposition techniques for
multivariate polynomials \cite{papachristodoulou2005tutorial}. 

The following notations are used:  $\mathbb{R}[\bm{x}]$ denotes the ring of all multivariate polynomials in a variable $\bm{x}$; $\sum[\bm{x}]$ is
used to represent the set of sum of squares polynomials over variables
$\bm{x}$, i.e., $\sum[\bm{x}] = \{p(\bm{x}) \in \mathbb{R}[\bm{x}] \mid  p(\bm{x}) =\sum_{i=1}^k q_i^2(\bm{x}), q_i\in\mathbb{R}[\bm{x}], i = 1,\ldots ,k\}$. In addition, we assume $\mathcal{X}=\{\bm{x}\in \mathbb{R}^n\mid h(\bm{x})\leq 0\}$ with $h(\bm{x}) \in \mathbb{R}[\bm{x}]$, and 
$\mathcal{X}_s=\{\bm{x}\in \mathbb{R}^n\mid g(\bm{x})\leq 0\}$ with $g(\bm{x}) \in \mathbb{R}[\bm{x}]$.

The semidefinite program for solving \eqref{fen2020} is:
\begin{equation*}
\begin{split}
&\textbf{SDP0}~~~~~~~~~~~~~\min  v(0,\bm{x}_0)\\
&\text{s.t.~}\\
&\begin{cases}
  -\mathcal{L}v(t,\bm{x}) +s_0(t,\bm{x})t(t-T)+s_1(t,\bm{x})h(\bm{x})-s_2(t,\bm{x})g(\bm{x})\in \sum[t,\bm{x}],\\
      -\frac{\partial v(t,\bm{x})}{\partial t}+ p(t,\bm{x})h(\bm{x}) +s_3(t,\bm{x})t(t-T) \in \sum[t,\bm{x}],\\
        v(t,\bm{x})+s_4(t,\bm{x})h(\bm{x})-s_5(t,\bm{x})g(\bm{x})+s_6(t,\bm{x})t(t-T)\in \sum[t,\bm{x}], \\
         v(t,\bm{x})-1+s_7(t,\bm{x})g(\bm{x})+s_8(t,\bm{x})t(t-T)\in \sum[t,\bm{x}]. 
\end{cases}
\end{split}
\end{equation*}
 The above semi-definite program is minimizing the objective via searching $(v(t,\bm{x}),p(t,\bm{x})\in \mathbb{R}[t,\bm{x}];s_i(t,\bm{x})\in \sum[t,\bm{x}],i=0,\ldots,8)$.

The semidefinite program for solving \eqref{lower_bound_condition2} is formulated below:
\begin{equation*}
\begin{split}
&\textbf{SDP1}~~~~~~~~~~~~~\min  obj\\
&\text{s.t.~}\begin{cases}
 \alpha v(t,\bm{x})+\beta -\mathcal{L}v(t,\bm{x}) +s_0(t,\bm{x})t(t-T)\\
 ~~~~~~~~~~~~~~+s_1(t,\bm{x})h(\bm{x})-s_2(t,\bm{x})g(\bm{x})\in \sum[t,\bm{x}],\\
        \alpha v(t,\bm{x})+\beta-\frac{\partial v(t,\bm{x})}{\partial t}+ p(t,\bm{x})h(\bm{x})\\
        ~~~~~~~~~~~~~~~~~~~~~~~~~~~~+s_3(t,\bm{x})t(t-T) \in \sum[t,\bm{x}],\\
          \alpha v(t,\bm{x})+\beta-\frac{\partial v(t,\bm{x})}{\partial t}+ q(t,\bm{x})g(\bm{x}) \\
          ~~~~~~~~~~~~~~~~~~~~~~~~~~~~+s_4(t,\bm{x})t(t-T) \in \sum[t,\bm{x}],\\
        v(T,\bm{x})+s_5(\bm{x})h(\bm{x})-s_6(\bm{x})g(\bm{x})\in \sum[\bm{x}] \\
         v(T,\bm{x})-1+r(\bm{x})g(\bm{x})\in \sum[\bm{x}],
\end{cases}
\end{split}
\end{equation*}
where $obj=v(0,\bm{x}_0)+\beta T~\text{if}~\alpha=0$ and $obj=
       e^{\alpha T}v(0,\bm{x}_0)+\frac{\beta}{\alpha}(e^{\alpha T}-1) ~\text{if}~ \alpha\neq 0$, and $\alpha$ is a user-specified value. The above semi-definite program is minimizing the objective via searching $(\beta;v(t,\bm{x})\in \mathbb{R}[t,\bm{x}];s_i(\bm{x})\in \sum[t,\bm{x}],i=0,\ldots,4;p(\bm{x}),q(\bm{x}),r(\bm{x})\in \mathbb{R}[\bm{x}];s_j(\bm{x}) \in \sum[\bm{x}],j=5,\ldots,6)$. 

The semidefinite program for solving \eqref{lower_bound_3} is formulated below: 
\begin{equation*}
\begin{split}
&\textbf{SDP2}~~~~~~~~~~~~~\min  obj\\
&\text{s.t.~}\begin{cases}
 \alpha v(\bm{x})+\beta -\mathcal{L}v(\bm{x})+s_1(\bm{x})h(\bm{x})-s_2(\bm{x})g(\bm{x})\in \sum[\bm{x}],\\
        \alpha v(\bm{x})+\beta+ p(\bm{x})h(\bm{x}) \in \sum[\bm{x}],\\
          \alpha v(\bm{x})+\beta+ q(\bm{x})g(\bm{x})  \in \sum[\bm{x}],\\
        v(\bm{x})+s_3(\bm{x})h(\bm{x})-s_4(\bm{x})g(\bm{x})\in \sum[\bm{x}], \\
         v(\bm{x})-1+r(\bm{x})g(\bm{x})\in \sum[\bm{x}],
\end{cases}
\end{split}
\end{equation*}
where $obj=v(\bm{x}_0)+\beta T~\text{if}~\alpha=0$ and $obj=
       e^{\alpha T}v(\bm{x}_0)+\frac{\beta}{\alpha}(e^{\alpha T}-1) ~\text{if}~ \alpha\neq 0$, and $\alpha$ is a user-specified value. This semi-definite program is minimizing the objective via searching $(\beta;s_i(\bm{x})\in \sum[\bm{x}],i=0,\ldots,4;v(\bm{x}),p(\bm{x}),q(\bm{x}),r(\bm{x})\in \mathbb{R}[\bm{x}])$.

    The semidefinite program for solving \eqref{upper_bound_2} is formulated below: 
\begin{equation*}
\begin{split}
&\textbf{SDP3}~~~~~~~~~~~~~\min  obj\\
&\text{s.t.~}\begin{cases}
  \mathcal{L}v(t,\bm{x})-\alpha v(t,\bm{x})-\beta +s_0(t,\bm{x})t(t-T)\\
 ~~~~~~~~~~~~+s_1(t,\bm{x})h(\bm{x})-s_2(t,\bm{x})g(\bm{x})\in \sum[t,\bm{x}],\\
        \frac{\partial v(t,\bm{x})}{\partial t}-\alpha v(t,\bm{x})-\beta+ p(t,\bm{x})h(\bm{x}) \\
        ~~~~~~~~~~~~~~~~~~~~~~~~~~~+s_3(t,\bm{x})t(t-T) \in \sum[t,\bm{x}],\\
          \frac{\partial v(t,\bm{x})}{\partial t}-\alpha v(t,\bm{x})-\beta+ q(t,\bm{x})g(\bm{x}) \\
          ~~~~~~~~~~~~~~~~~~~~~~~~~~~+s_4(t,\bm{x})t(t-T) \in \sum[t,\bm{x}],\\
          1+\frac{\partial w(t,\bm{x})}{\partial t}-v(t,\bm{x})+r(t,\bm{x})h(\bm{x})\\
          ~~~~~~~~~~~~~~~~~~~~~~~~~~~+s_5(t,\bm{x})t(t-T) \in \sum[t,\bm{x}],\\
          \mathcal{L}w(t,\bm{x})-v(t,\bm{x})+ s_6(t,\bm{x})t(t-T)\\
 ~~~~~~~~~~~~+s_7(t,\bm{x})h(\bm{x})-s_8(t,\bm{x})g(\bm{x})\in \sum[t,\bm{x}],\\
       \frac{\partial w(t,\bm{x})}{\partial t}-v(t,\bm{x})+s_9(t,\bm{x})t(t-T)\\
        ~~~~~~~~~~~~~~~~~~~~~~~~~~~~~~~~+l(t,\bm{x})h(\bm{x})\in \sum[t,\bm{x}], 
\end{cases}
\end{split}
\end{equation*}
where $obj=\frac{(\frac{1}{\alpha}v(0,\bm{x}_0)+\frac{\beta}{\alpha^2})(e^{\alpha T}-1)-\frac{\beta}{\alpha}T}{T} -\frac{2M}{T}~\text{if}~ \alpha\neq 0$ and $obj=v(0,\bm{x}_0)+\frac{1}{2}\beta T-\frac{2M}{T}~\text{if}~\alpha=0$, and $\alpha$ is a user-specified value. The above semi-definite program is minimizing the objective via searching $(\beta;s_i(t,\bm{x})\in \sum[t,\bm{x}],i=0,\ldots,9;v(t,\bm{x}),w(t,\bm{x}),p(t,\bm{x}),q(t,\bm{x}),r(t,\bm{x}), l(t,\bm{x})\in \mathbb{R}[t,\bm{x}]$. 

   The semidefinite program for solving \eqref{upper_bound_3} is formulated below: 
\begin{equation*}
\begin{split}
&\textbf{SDP4}~~~~~~~~~~~~~\min  obj\\
&\text{s.t.~}\begin{cases}
  \mathcal{L}v(\bm{x})-\alpha v(\bm{x})-\beta +s_0(\bm{x})h(\bm{x})-s_1(\bm{x})g(\bm{x})\in \sum[\bm{x}],\\
        -\alpha v(\bm{x})-\beta+ p(\bm{x})h(\bm{x}) \in \sum[\bm{x}],\\
          -\alpha v(\bm{x})-\beta+ q(\bm{x})g(\bm{x}) \in \sum[\bm{x}],\\
          1-v(\bm{x})+r(\bm{x})h(\bm{x}) \in \sum[\bm{x}],\\
          \mathcal{L}w(\bm{x})-v(\bm{x})+s_2(\bm{x})h(\bm{x})-s_3(\bm{x})g(\bm{x})\in \sum[\bm{x}],\\
       -v(\bm{x})+l(\bm{x})h(\bm{x})\in \sum[\bm{x}],
\end{cases}
\end{split}
\end{equation*}
where $obj=\frac{(\frac{1}{\alpha}v(0,\bm{x}_0)+\frac{\beta}{\alpha^2})(e^{\alpha T}-1)-\frac{\beta}{\alpha}T}{T} -\frac{2M}{T}~\text{if}~ \alpha\neq 0$ and $obj=v(0,\bm{x}_0)+\frac{1}{2}\beta T-\frac{2M}{T}~\text{if}~\alpha=0$, and $\alpha$ is a user-specified value. The above semi-definite program is minimizing the objective via searching $(\beta;s_i(\bm{x})\in \sum[\bm{x}],i=0,\ldots,3;v(\bm{x}),w(\bm{x}),p(\bm{x}),q(\bm{x}),r(\bm{x}), l(\bm{x})\in \mathbb{R}[\bm{x}]$.   

  The semidefinite program for solving \eqref{lower_bound_condition2II} is formulated below: 

  \begin{equation*}
\begin{split}
&\textbf{SDP5}~~~~~~~~~~~~~\min  obj\\
&\text{s.t.~}\begin{cases}
   \alpha v(t,\bm{x})+\beta-\mathcal{L}v(t,\bm{x})+s_0(t,\bm{x})h(\bm{x})+s_1(t,\bm{x})t(t-T)\in \sum[t,\bm{x}],\\
        \alpha v(t,\bm{x})+\beta-\frac{\partial v(t,\bm{x})}{\partial t}+ p(t,\bm{x})h(\bm{x})+s_2(t,\bm{x})t(t-T) \in \sum[t,\bm{x}],\\
          v(T,\bm{x})+s_3(\bm{x})h(\bm{x})-s_4(\bm{x})g(\bm{x})\in \sum[\bm{x}], \\
         v(T,\bm{x})-1+s_5(\bm{x})g(\bm{x})\in \sum[\bm{x}],
\end{cases}
\end{split}
\end{equation*}
where $obj=e^{\alpha T}v(0,\bm{x}_0)+\frac{\beta}{\alpha}(e^{\alpha T}-1)~\text{if}~ \alpha\neq 0$ and $obj=v(0,\bm{x}_0)+\beta T~\text{if}~\alpha=0$, and $\alpha$ is a user-specified value. The above semi-definite program is minimizing the objective via searching $(\beta;s_i(t,\bm{x})\in \sum[t,\bm{x}],i=0,\ldots,2;s_j(\bm{x})\in \sum[\bm{x}],j=3,\ldots,5;v(t,\bm{x}),p(t,\bm{x})\in \mathbb{R}[t,\bm{x}]$.

  The semidefinite program for solving \eqref{lower_bound_3II} is formulated below: 

  \begin{equation*}
\begin{split}
&\textbf{SDP6}~~~~~~~~~~~~~\min  obj\\
&\text{s.t.~}\begin{cases}
   \alpha v(\bm{x})+\beta-\mathcal{L}v(\bm{x})+s_0(\bm{x})h(\bm{x})\in \sum[\bm{x}],\\
        \alpha v(\bm{x})+\beta+ p(\bm{x})h(\bm{x})\in \sum[t,\bm{x}],\\
          v(\bm{x})+s_1(\bm{x})h(\bm{x})-s_2(\bm{x})g(\bm{x})\in \sum[\bm{x}], \\
         v(\bm{x})-1+s_3(\bm{x})g(\bm{x})\in \sum[\bm{x}], 
\end{cases}
\end{split}
\end{equation*}
where $obj=e^{\alpha T}v(\bm{x}_0)+\frac{\beta}{\alpha}(e^{\alpha T}-1)~\text{if}~ \alpha\neq 0$ and $obj=v(\bm{x}_0)+\beta T~\text{if}~\alpha=0$, and $\alpha$ is a user-specified value. The above semi-definite program is minimizing the objective via searching $(\beta;s_i(\bm{x})\in \sum[\bm{x}],j=0,\ldots,3;v(\bm{x}),p(\bm{x})\in \mathbb{R}[t,\bm{x}]$. 

  The semidefinite program for solving \eqref{upper_bound_2II} is formulated below: 
 \begin{equation*}
\begin{split}
&\textbf{SDP7}~~~~~~~~~~~~~\min  obj\\
&\text{s.t.~}\begin{cases}
   -\alpha v(t,\bm{x})-\beta+\mathcal{L}v(t,\bm{x})+s_0(t,\bm{x})h(\bm{x})+s_1(t,\bm{x})t(t-T)\in \sum[t,\bm{x}],\\
        -\alpha v(t,\bm{x})-\beta+\frac{\partial v(t,\bm{x})}{\partial t}+ p(t,\bm{x})h(\bm{x}) +s_2(t,\bm{x})t(t-T) \in \sum[t,\bm{x}],\\
          -v(T,\bm{x})+s_3(\bm{x})h(\bm{x})-s_4(\bm{x})g(\bm{x})\in \sum[\bm{x}], \\
         -v(T,\bm{x})+1+s_5(\bm{x})g(\bm{x})\in \sum[\bm{x}], 
\end{cases}
\end{split}
\end{equation*}
where $obj=-(e^{\alpha T}v(0,\bm{x}_0)+\frac{\beta}{\alpha}(e^{\alpha T}-1))~\text{if}~ \alpha\neq 0$ and $obj=-(v(0,\bm{x}_0)+\beta T)~\text{if}~\alpha=0$, and $\alpha$ is a user-specified value. The above semi-definite program is minimizing the objective via searching $(\beta;s_i(t,\bm{x})\in \sum[t,\bm{x}],i=0,\ldots,2;s_j(\bm{x})\in \sum[\bm{x}],j=3,\ldots,5;v(t,\bm{x}),p(t,\bm{x})\in \mathbb{R}[t,\bm{x}]$.

\subsection{Connections to Deterministic Reach-avoid Problems}
\label{app:deter}

This appendix explores the connections between the stochastic barrier functions proposed in the main text and their counterparts in deterministic systems. When the diffusion term $\boldsymbol{\sigma}(\cdot)$ is identically zero, the stochastic system (1) reduces to an ordinary differential equation. The following remarks detail how the barrier conditions simplify and can be used to provide guarantees for deterministic reach-avoid problems.

 \begin{remark}[Deterministic Reach-Avoid via Stochastic Barrier Functions I]
\label{re3}
The barrier function for certifying reach-avoid analysis over the time horizon $[0,T]$ (i.e., reach the set $\mathcal{X}_r$ at some time instant $t\in [0,T]$ while staying within the set $\mathcal{X}$ before $t$) for deterministic systems can be retrieved from the one satisfying \eqref{upper_bound_1}. 

When $\bm{\sigma}(\,\cdot\,)\equiv \bm{0}$, stochastic system \eqref{SDE} degenerates into a deterministic system 
\begin{equation}
    \label{ode}
    \frac{d \bm{x}}{d t}=\bm{b}(\bm{x}),
\end{equation}
where $\bm{X}_{\bm{x}_0}(\,\cdot\,)\colon [0,T^{\bm{x}_0})\rightarrow \mathbb{R}^n$ is the solution  with the initial state $\bm{x}_0$ at $t=0$. 

If there exist a barrier function $v(t,\bm{x})\colon [0,T]\times \mathbb{R}^n\rightarrow \mathbb{R}$ and a function $w(t,\bm{x})\colon [0,T]\times \mathbb{R}^n \rightarrow \mathbb{R}$ with $\sup_{(t,\bm{x})\in [0,T]\times \overline{\mathcal{X}}}|w(t,\bm{x})|\leq M$ that are continuously differentiable over $t$ and $\bm{x}$, satisfying condition \eqref{upper_bound_1}, then the system described by equation \eqref{ode}, when initiated from the set $\{\,\bm{x}\in \mathcal{X}\setminus \mathcal{X}_s\mid v(0,\bm{x})>\frac{2M}{T}\,\}$, will reach the set $\mathcal{X}_s$ within the time interval $[0,T]$, while remaining within the set $\mathcal{X}$ before the first encounter with $\mathcal{X}_s$. A proof of this conclusion is shown here: we first demonstrate that if system \eqref{ode} is initiated from the set $\{\,\bm{x}\in \mathcal{X}\setminus \mathcal{X}_s\mid v(0,\bm{x})>\frac{2M}{T}\,\}$, within the time interval $[0,T]$, it remains within the set $\mathcal{X}$ if it does not reach the set $\mathcal{X}_s$. Assume it is not true. Then there exists $\tau\in [0,T]$ such that $\bm{X}_{\bm{x}_0}(\tau) \in \partial \mathcal{X}$ and $\bm{X}_{\bm{x}_0}(t) \in \mathcal{X}\setminus \mathcal{X}_s$ for $t\in [0,\tau)$. Since $\mathcal{L}v(t,\bm{x})\geq 0$ for $t\in [0,T]$ and $\bm{x}\in \mathcal{X}\setminus \mathcal{X}_s$, we have \[v(\tau,\bm{X}_{\bm{x}_0}(\tau))\geq v(0,\bm{x}_0)>\frac{2M}{T}.\] Also, according to $v(t,\bm{x})\leq \mathcal{L}w(t,\bm{x})$ for $(t,\bm{x})\in [0,T]\times \mathcal{X}\setminus \mathcal{X}_s$, we have \[\frac{2M}{T} \tau< \int_{0}^{\tau} v(t,\bm{x})dt\leq w(\tau,\bm{X}_{\bm{x}_0}(\tau))-w(0,\bm{x}_0)\] and thus $w(\tau,\bm{X}_{\bm{x}_0}(\tau))> \frac{2M}{T}\tau + w(0,\bm{x}_0)$. Further, since $\frac{\partial v(t,\bm{x})}{\partial t}\geq 0$ and $v(t,\bm{x})\leq \frac{\partial w(t,\bm{x})}{\partial t}$ for $(t,\bm{x})\in [0,T]\times \partial \mathcal{X}$, we have $v(\tau,\bm{X}_{\bm{x}_0}(\tau))(T-\tau)\leq \int_{\tau}^{T}v(t,\bm{X}_{\bm{x}_0}(\tau))dt \leq \int_{\tau}^{T}\frac{\partial w(t,\bm{X}_{\bm{x}_0}(\tau))}{\partial t}dt=w(T,\bm{X}_{\bm{x}_0}(\tau))-w(\tau,\bm{X}_{\bm{x}_0}(\tau))$. Thus, 
\begin{equation}
\begin{split}
\frac{2M}{T}(T-\tau)&<v(\tau,\bm{X}_{\bm{x}_0}(\tau))(T-\tau)\\
&\leq w(T,\bm{X}_{\bm{x}_0}(\tau))-\frac{2M}{T}\tau - w(0,\bm{x}_0)
\end{split}
\end{equation}
holds and consequently, $2M<w(T,\bm{X}_{\bm{x}_0}(\tau))- w(0,\bm{x}_0) \leq 2M$, which is a contradiction. Therefore, the system \eqref{ode}, when initiated from the set $\{\,\bm{x}\in \mathcal{X}\setminus \mathcal{X}_s\mid v(0,\bm{x})>\frac{2M}{T}\,\}$, will not reach the set $\partial \mathcal{X}$ within the time interval $[0,T]$ if it does not enter the set $\mathcal{X}_s$. Next, we show that the system \eqref{ode}, when initiated from the set $\{\,\bm{x}\in \mathcal{X}\setminus \mathcal{X}_s\mid v(0,\bm{x})>\frac{2M}{T}\,\}$, will reach the set $\mathcal{X}_s$ within the time interval $[0,T]$.  Assume that this conclusion does not hold, which indicates that the system when initiated from the set $\{\,\bm{x}\in \mathcal{X}\setminus \mathcal{X}_s\mid v(0,\bm{x})>\frac{2M}{T}\,\}$, will stay within the set $\mathcal{X}\setminus \mathcal{X}_s$ over the time horizon $[0,T]$. From $v(t,\bm{x})\leq \mathcal{L}w(t,\bm{x})$ for $(t,\bm{x})\in [0,T]\times \mathcal{X}\setminus \mathcal{X}_s$, we can obtain $\frac{2M}{T} T< \int_{0}^T v(t,\bm{X}_{\bm{x}_0}(t))dt \leq w(T,\bm{X}_{\bm{x}_0}(T))- w(0,\bm{x}_0)\leq 2M$, which is a contradiction. Therefore, the system \eqref{ode}, when initiated from the set $\{\,\bm{x}\in \mathcal{X}\setminus \mathcal{X}_s\mid v(0,\bm{x})>\frac{2M}{T}\,\}$, will reach the set $\mathcal{X}_s$ within the time interval $[0,T]$ and finally the conclusion holds. 

\end{remark}

\begin{remark}[Deterministic Reach-Avoid via Stochastic Barrier Functions II]
\label{re4}
If there exist a barrier function $v(t,\bm{x})\colon [0,T]\times \mathbb{R}^n\rightarrow \mathbb{R}$ and a function $w(t,\bm{x})\colon [0,T]\times \mathbb{R}^n \rightarrow \mathbb{R}$ with $\sup_{(t,\bm{x})\in [0,T]\times \overline{\mathcal{X}}}|w(t,\bm{x})|\leq M$ that are continuously differentiable over $t$ and $\bm{x}$, satisfying condition \eqref{upper_bound_2}, then the system \eqref{ode}, when initiated from the set $\{\,\bm{x}\in \mathcal{X}\setminus \mathcal{X}_s\mid \frac{(\frac{1}{\alpha}v(0,\bm{x})+\frac{\beta}{\alpha^2})(e^{\alpha T}-1)-\frac{\beta}{\alpha}T}{T} -\frac{2M}{T}>0\,\}$ (if $\alpha\neq 0$) or 
  $\{\,\bm{x}\in \mathcal{X}\setminus \mathcal{X}_s\mid v(0,\bm{x})+\frac{1}{2}\beta T-\frac{2M}{T}>0\,\}$ (if $\alpha=0$), will reach the set $\mathcal{X}_s$ within the time interval $[0,T]$, while remaining within the set $\mathcal{X}$ before the first encounter with $\mathcal{X}_s$.
\end{remark}

\begin{remark}[Deterministic Reach-Avoid via Stochastic Barrier Functions III]
Similar to the time-dependent case, the barrier condition \eqref{upper_bound_3} also provides a guarantee for deterministic systems ($\boldsymbol{\sigma}(\cdot) \equiv \boldsymbol{0}$). Specifically, the system \eqref{ode}, when initiated from a state where the lower bound in Corollary \eqref{time-in-H} is positive, will reach the set $\mathcal{X}_s$ within the time interval $[0,T]$ while remaining within the set $\mathcal{X}$ until the first encounter with $\mathcal{X}_s$.
\end{remark}

\begin{remark}[Deterministic Reach at Time $T$ via Stochastic Barrier Functions I]

If there exists a barrier function $v(t,\bm{x})\colon [0,T]\times \mathbb{R}^n\rightarrow \mathbb{R}$ that is continuously differentiable over $t$ and $\bm{x}$, satisfying condition \eqref{upper_bound_1II}, then the system described by equation \eqref{ode}, when initiated from the set $\{\,\bm{x}\in \mathcal{X}\setminus \mathcal{X}_s\mid v(0,\bm{x})>0\,\}$, will reach the set $\mathcal{X}_s$ at $t=T$, while remaining within the set $\mathcal{X}$ before $t=T$. This problem has been studied in \cite{xue2019inner}. A proof of this conclusion is shown here: we first demonstrate that if system \eqref{ode} is initiated from the set $\{\,\bm{x}\in \mathcal{X}\setminus \mathcal{X}_s\mid v(0,\bm{x})>0\,\}$, it will stay within the set $\mathcal{X}$ over the time horizon $[0,T]$. Assume this is not true. Then there exists $\tau\in [0,T]$ such that $\bm{X}_{\bm{x}_0}(\tau) \in \partial \mathcal{X}$ and $\bm{X}_{\bm{x}_0}(t) \in \mathcal{X}$ for $t\in [0,\tau)$. Since $\mathcal{L}v(t,\bm{x})\geq 0$ for $t\in [0,T]$ and $\bm{x}\in \mathcal{X}$, we have $v(\tau,\bm{X}_{\bm{x}_0}(\tau))\geq v(0,\bm{x}_0)>0$. Further, since $\frac{\partial v(t,\bm{x})}{\partial t}\geq 0$ for $(t,\bm{x})\in [0,T]\times \partial \mathcal{X}$, we have \[v(T,\bm{X}_{\bm{x}_0}(\tau))\geq v(\tau,\bm{X}_{\bm{x}_0}(\tau)) > 0,\] which contradicts $v(T,\bm{x})\leq 1_{\mathcal{X}_s}(\bm{x}), \forall \bm{x}\in \overline{\mathcal{X}}$. Therefore, the system \eqref{ode}, when initiated from the set $\{\,\bm{x}\in \mathcal{X}_s\mid v(0,\bm{x})>0\,\}$, stays within the set $\mathcal{X}$ over the time horizon $[0,T]$. Next, we show that the system \eqref{ode}, when initiated from the set $\{\,\bm{x}\in \mathcal{X}\setminus \mathcal{X}_s\mid v(0,\bm{x})>0\,\}$, will reach the set $\mathcal{X}_s$ at $t=T$.  Assume that this conclusion does not hold, which indicates that the system, when initiated from the set $\{\,\bm{x}\in \mathcal{X}\setminus \mathcal{X}_s\mid v(0,\bm{x})>0\,\}$, will stay within the set $\mathcal{X}\setminus \mathcal{X}_s$ over the time horizon $[0,T]$. From $\mathcal{L}v(t,\bm{x})\geq 0$ for $(t,\bm{x})\in [0,T]\times \mathcal{X}\setminus \mathcal{X}_s$, we can obtain $v(T,\bm{X}_{\bm{x}_0}(T))>0$, which contradicts $v(T,\bm{x})\leq 1_{\mathcal{X}_s}(\bm{x}), \forall\bm{x}\in \overline{\mathcal{X}}$. Therefore, the system \eqref{ode}, when initiated from the set $\{\,\bm{x}\in \mathcal{X}_s\mid v(0,\bm{x})>0\,\}$, will reach the set $\mathcal{X}_s$ at $t=T$ and finally the conclusion holds. 
\end{remark}

\begin{remark}[Deterministic Reach at Time $T$ via Stochastic Barrier Functions II]
\label{re7}
If there exists a barrier function $v(t,\bm{x})\colon [0,T]\times \mathbb{R}^n\rightarrow \mathbb{R}$ that is continuously differentiable over $t$ and $\bm{x}$, satisfying condition \eqref{upper_bound_2II}, then the system \eqref{ode}, when initiated from the set $\{\,\bm{x}\in \mathcal{X}\setminus \mathcal{X}_s\mid e^{\alpha T}v(0,\bm{x})+\frac{\beta}{\alpha}(e^{\alpha T}-1)>0\,\}$ (if $\alpha\neq 0$) or 
  $\{\,\bm{x}\in \mathcal{X}\setminus \mathcal{X}_s\mid v(0,\bm{x})+\beta T>0\,\}$ (if $\alpha=0$), will reach the set $\mathcal{X}_s$ at $t=T$, while remaining within the set $\mathcal{X}$ before the time instant $T$.
\end{remark}

\subsection{Explanation on the Ineffectiveness of Condition~\eqref{upper_bound_3II}}
\label{section:expla}
Condition \eqref{upper_bound_3II} is useless in determining lower bounds of the probability $\mathbb{P}_{\bm{x}_0}^{T}$. We give a brief explanation here. Firstly, we observe that $\alpha$ should not be zero. If $\alpha=0$,  we have $\mathbb{P}_{\bm{x}_0}^{T}\geq v(\bm{x})+\beta T$ for $\bm{x}\in \mathcal{X}\setminus \mathcal{X}_s$. Further, $0\geq \alpha v(\bm{x})+\beta, \forall \bm{x}\in \partial \mathcal{X}$ implies $\beta\leq 0$. Moreover, since $v(\bm{x})\leq 0$ over $\bm{x}\in \overline{\mathcal{X}}\setminus \mathcal{X}_s$, we conclude that 
$v(\bm{x})+\beta T\leq 0$ is a useless lower bound. Secondly, if $\alpha<0$, since $v(\bm{x})\leq 0$ over $\bm{x}\in \overline{\mathcal{X}}\setminus \mathcal{X}_s$ and $0\geq \alpha v(\bm{x})+\beta, \forall \bm{x}\in \partial \mathcal{X}$, we have $\beta\leq 0$. Thus, $e^{\alpha T}v(\bm{x})+\frac{\beta}{\alpha}(e^{\alpha T}-1)\leq 0$ is a useless lower bound of $\mathbb{P}_{\bm{x}}^{T}$ for $\bm{x}\in \mathcal{X}\setminus \mathcal{X}_s$. Thirdly, if $\alpha>0$ and $\beta\leq 0$, we have $e^{\alpha T}v(\bm{x}_0)+\frac{\beta}{\alpha}(e^{\alpha T}-1)\leq 0$ is useless lower bound of $\mathbb{P}_{\bm{x}}^{T}$ for $\bm{x}\in \mathcal{X}\setminus \mathcal{X}_s$. Fourthly, if $\alpha>0$ and $\beta> 0$,  we have $e^{\alpha T}v(\bm{x}_0)+\frac{\beta}{\alpha}(e^{\alpha T}-1)= e^{\alpha T}(v(\bm{x}_0)+\frac{\beta}{\alpha})-\frac{\beta}{\alpha} \leq 0$ is useless lower bound of $\mathbb{P}_{\bm{x}}^{T}$ for $\bm{x}\in \{\mathcal{X}\setminus \mathcal{X}_s\mid v(\bm{x})+\frac{\beta}{\alpha}\leq 0\}$. Lastly, we show that $\bm{x}\in \{\mathcal{X}\setminus \mathcal{X}_s\mid v(\bm{x})+\frac{\beta}{\alpha}>0\}=\emptyset$, where $\alpha>0$ and $\beta> 0$. Assume there exists $\bm{x}_0\in \{\mathcal{X}\setminus \mathcal{X}_s\mid v(\bm{x})+\frac{\beta}{\alpha}> 0\}$, where $\alpha>0$ and $\beta> 0$. Following the proof of Theorem \ref{coro5II}, we can obtain 
\begin{equation*}
\begin{split}
\mathbb{P}(\widetilde{\bm{X}}_{\bm{x}_0}^{\bm{w}}(T)\in \mathcal{X}_s)&=\mathbb{E}[1_{\mathcal{X}_s}(\widetilde{\bm{X}}_{\bm{x}_0}^{\bm{w}}(T))]\geq  \mathbb{E}[v(\widetilde{\bm{X}}_{\bm{x}_0}^{\bm{w}}(T))]\\
&\geq e^{\alpha T}v(\bm{x}_0)+\frac{\beta}{\alpha}(e^{\alpha T}-1)=e^{\alpha T}(v(\bm{x}_0)+\frac{\beta}{\alpha})-\frac{\beta}{\alpha}.
  \end{split}
\end{equation*}
Also, since $\lim_{T\rightarrow +\infty}e^{\alpha T}(v(\bm{x}_0)+\frac{\beta}{\alpha})-\frac{\beta}{\alpha}=+\infty$,  $\lim_{T\rightarrow +\infty}\mathbb{E}[v(\widetilde{\bm{X}}_{\bm{x}_0}^{\bm{w}}(T))]=+\infty$ holds, which contradicts $v(\bm{x})\leq 1_{\mathcal{X}_s}(\bm{x}), \forall \bm{x}\in \overline{\mathcal{X}}$. Therefore, $\bm{x}\in \{\mathcal{X}\setminus \mathcal{X}_s\mid v(\bm{x})+\frac{\beta}{\alpha}>0\}=\emptyset$, where $\alpha>0$ and $\beta> 0$. In summary, condition \eqref{upper_bound_3II} is useless in determining lower bounds of the probability $\mathbb{P}_{\bm{x}_0}^{T}$.
\end{document}